\documentclass{bmcart}
\usepackage{amsthm,amsmath}
\usepackage[T1]{fontenc}
\usepackage[utf8]{inputenc}
\usepackage{amsmath,amsfonts,amssymb,graphics,graphicx,fancyhdr,subfigure,float,fancybox,url}

\usepackage[linesnumbered,ruled]{algorithm2e}
\SetKwInOut{Parameter}{Parameter}

\startlocaldefs

\usepackage{enumitem}
\usepackage{xcolor}

\newtheorem{prop}{Proposition}

\newcommand{ \bm }[1]{ \mbox{\bf {#1}} }

\newcommand{\by}{\bm{y}}

\newcommand{\beginsupplement}{%
        \setcounter{table}{0}
        \renewcommand{\thetable}{S\arabic{table}}%
        \setcounter{figure}{0}
        \renewcommand{\thefigure}{S\arabic{figure}}%
     }

\usepackage[
bookmarks=false,
breaklinks=true,
colorlinks=true,
linkcolor=black,
citecolor=black,
urlcolor=black,
pdfpagelayout=SinglePage
]{hyperref}

\endlocaldefs

\begin{document}

\title{Jaccard/Tanimoto similarity test \\ and estimation methods}

\author{{\small Neo Christopher Chung$^\textnormal{1,*}$, B\l{}a{\.z}ej Miasojedow$^\textnormal{2}$, Micha\l{} Startek$^\textnormal{1}$, Anna Gambin$^\textnormal{1}$}}
\date{}
\maketitle 

\vspace*{-5mm}
\begin{center}
{\small $^\textnormal{1}$Institute of Informatics, University of Warsaw \\
$^\textnormal{2}$Institute of Mathematics, Polish Academy of Sciences \\
$^\textnormal{*}$nchchung@gmail.com}
\end{center} 
\vspace*{10mm}

A survey of presences and absences of specific species across multiple biogeographic units (or bioregions) are used in a broad area of biological studies from ecology to microbiology. Using binary presence-absence data, we evaluate species co-occurrences that help elucidate relationships among organisms and environments. To summarize similarity between occurrences of species, we routinely use the Jaccard/Tanimoto coefficient, which is the ratio of their intersection to their union. It is natural, then, to identify statistically significant Jaccard/Tanimoto coefficients, which suggest non-random co-occurrences of species. However, statistical hypothesis testing using this similarity coefficient has been seldom used or studied.

%\parttitle{Results} %if any
We introduce a hypothesis test for similarity for biological presence-absence data, using the Jaccard/Tanimoto coefficient. Several key improvements are presented including unbiased estimation of expectation and centered Jaccard/Tanimoto coefficients, that account for occurrence probabilities. The exact and asymptotic solutions are derived. To overcome a computational burden due to high-dimensionality, we propose the bootstrap and measurement concentration algorithms to efficiently estimate statistical significance of binary similarity. Simulation studies demonstrate that our proposed methods produce accurate p-values and false discovery rates. The proposed estimation methods are orders of magnitude faster than the exact solution, particularly with an increasing dimensionality. We showcase their application in evaluating co-occurrences of bird species in 28 islands of Vanuatu and fish species in 3347 freshwater habitats in France. The proposed methods are implemented in a R package called {\tt jaccard} (\url{https://cran.r-project.org/package=jaccard}).

%\parttitle{Conclusion} %if any
We introduce a suite of statistical methods for the Jaccard/Tanimoto similarity coefficient, that enable straightforward incorporation of probabilistic measures in analysis for species co-occurrences. Due to their generality, the proposed methods and implementations are applicable to a wide range of binary data arising from genomics, biochemistry, and other areas of science.  
\noindent {\small Funding: Narodowe Centrum Nauki 2014/12/W/ST5/00592 and 2016/23/D/ST6/03613}
\clearpage
%%%%%%%%%%%%%%%%
%% Background %%
%%
\section*{Background}
%\textcolor{red}{This is a new sentence in the revised manuscript.}\\ 
Analysis of species co-occurrences helps us understand ecological and biological relationships among species. Essentially, the presence (1) and absence (0) of species are surveyed in multiple biogeographic units (or bioregions) using fieldwork, imaging, sequencing, and other techniques. Then, the Jaccard/Tanimoto coefficient is one of the most fundamental and popular similarity measures to compare such biological presence-absence data. Given two presence-absence vectors $\by_i$ and $\by_j$ of length $m$ that represent two different species, the Jaccard/Tanimoto similarity coefficient is the ratio of their intersection to their union, $T(\by_i,\by_j) = \by_{i} \cap \by_{j} / \by_{i} \cup \by_{j}$ \cite{Jaccard1912,Tanimoto1958}. This quantification of overlaps allows us to quantify co-existence of species \cite{Birks1987, Jackson1992, RealVargas1996, Manly2006}. However, the Jaccard/Tanimoto coefficient lacks probabilistic interpretations or statistical error controls. Surprisingly, its statistical properties, hypothesis testing, and estimation methods for p-values have been inadequately studied. Here, we present a rigorous statistical test evaluating the similarity in presence-absence data, derive exact and asymptotic solutions, and introduce efficient estimation methods for significance of the Jaccard/Tanimoto similarity coefficient. 

Generally, analysis of co-occurrences enables us to distinguish generalist species that survive in a broad range of environments from specialists that only thrive in a few localities \cite{Davies1993, Townsend2002}. Alternatively, similarity between two localities -- how two biogeographic units share an overlapping set of species -- sheds light on the beta diversity that may arise from ecological processes over time \cite{Whittaker1960,Harrison1992,Koleff2003}. There has been a long standing discussion on how to conduct association analysis for occurrences of species, including appropriate null models and evaluation techniques \cite{Connor1979, Diamond1982, Gilpin1982, Wilson1987, Manly1995, Sanderson1998}. There are also specialized probabilistic approaches, including metrics related to the Jaccard/Tanimoto coefficient \cite{Ellwood2009, Chase2011, Fridley2007, Araujo2013}. Yet, these studies rarely utilized statistical significance. Therefore, we investigated a hypothesis test using the Jaccard/Tanimoto coefficient that underlies or accompanies most of such association analyses.

The Jaccard/Tanimoto coefficient measuring similarity between two species has long been used to evaluate co-occurrences between species or between biogeographic units \cite{Baroni-Urbani1976, Baroni-Urbani1979, Birks1987, Jackson1992, RealVargas1996, Veech2013}. Pioneering early works on probabilistic treatment of the Jaccard/Tanimoto coefficient assume that the probability of species occurrences is 0.5 \cite{Baroni-Urbani1976, Baroni-Urbani1979, RealVargas1996}. These can be seen as special cases of our methods where both probabilities of $\by_i$ and $\by_j$ are set to 0.5. Recently, \cite{Veech2013} and \cite{Griffith2016} proposed estimating p-values with combinatorics and hypergeometric distributions, respectively. We found that they are inaccurate. To provide a comprehensive statistical treatment, we have developed a suite of methods and estimation techniques for rigorously testing similarity between presence-absence data. 

We derive a hypothesis test from the first principles using the Jaccard/Tanimoto coefficient. In the process, we propose an unbiased estimation of expectation and a centered Jaccard/Tanimoto coefficient that accounts for different probabilities of species occurrences. The negative and positive values of the centered Jaccard/Tanimoto coefficient naturally correspond to negative and positive association. We introduce an exact distribution of Jaccard/Tanimoto similarity coefficients under independence that is shown to provide accurate p-values. Because the exact solution for a large $m$ is computationally expensive, we have developed two efficient and accurate estimation algorithms. We demonstrate their remarkable accuracy and computational efficiency in comprehensive simulation studies, where p-values and false discovery rates (FDRs) are evaluated. As applications, we evaluated co-occurrences of bird species from $m=28$ islands of Vanuatu and of fish species from $m=3347$ freshwater habitats in France. 

All proposed methods are implemented in a statistical programming language R \cite{R}, available on the Comprehensive R Archive Network (\url{https://cran.r-project.org/package=jaccard}). We additionally provide an interactive web app (\url{https://nnnn.shinyapps.io/jaccard}). The implementations are efficient and general, such that the {\tt jaccard} package can rigorously test similarity between binary data arising from genomics, biochemistry, and others.

\section*{Methods}
\subsection*{Statistical Model and Test}

Quantitative comparison of presence-absence data in ecology and biology plays a crucial role in evaluating species co-existences, biodiversities, and ecosystems. In particular, one may be interested in comparing how species are co-occurring in biogeographic units or how biogeographic units are occupied by certain species. Note that species are used generally to indicate groups of organisms under investigations, such as operational taxonomic units (OTUs); similarly, biogeographic units or bioregions could be distinct survey areas, islands, or habitats. We are interested in statistically testing similarity between a pair of presence-absence data.

Given two presence-absence vectors $\by_i$ and $\by_j$ of length $m$, we are interested in inferring whether they are significantly related. Consider presence (1) and absence (0) of two species are recorded at $m$ biogeographic units. We measure their similarity by the ratio of their intersection to their union, $T(\by_i,\by_j) = \by_{i} \cap \by_{j} / \by_{i} \cup \by_{j}$. This is well known as the Jaccard/Tanimoto index or similarity coefficient \cite{Jaccard1912,Tanimoto1958}. In order to utilize the Jaccard/Tanimoto similarity coefficient in a statistically rigorous manner, we propose a family of methods and algorithms (Figure 1).

Under the null model of independence, $\by_i$ and $\by_j$ are assumed to be independent and identically distributed (i.i.d.). They are modeled by a Bernoulli distribution, with corresponding occurrence (i.e., success) probabilities $p_i$ and $p_j$ $\in [0,1]$. Specifically, for $k = 1, \ldots, m$, $\by_{i,k} \sim_{i.i.d.} \textnormal{Bernoulli}(p_i)$ and $\by_{j,k} \sim_{i.i.d.} \textnormal{Bernoulli}(p_j)$. Because this conventional definition is undefined if both binary vectors contain only zeros such that $\by_{i} \cup \by_{j} = 0$, we refine the definition of Jaccard/Tanimoto coefficient
\begin{equation} \label{def:tanimoto}
T(\by_i,\by_j)=
\begin{cases}
\frac{ \by_i \cap \by_j }{ \by_i \cup \by_j } & \text{if }  \by_i \cup \by_j \neq 0 \\
\frac{p_ip_j}{p_i+p_j-p_ip_j} & \text{otherwise.}
\end{cases}
\end{equation}

Following the definition of Jaccard/Tanimoto similarity coefficient in Eq. \eqref{def:tanimoto}, we derive its expected value $\mathbb E[T(\by_i,\by_j)] =\frac{p_i p_j}{p_i + p_j - p_i p_j}$. Substantial deviation from the expected value signifies similarity. Note that the Jaccard/Tanimoto coefficient can also be defined in terms of a multinomial distribution with four categories and $m$ trials (for example, representing $m$ biogeographic units). Four categories arising from presence-absence data are $N_1=\by_{i} \cap \by_{j}$,
$N_2= \by_{i} \cap(1- \by_{j})$, $N_3= (1-\by_{i}) \cap \by_{j}$ and
$N_4=m -N_1-N_2-N_3$. From $p_i$ and $p_j$, probabilities of those four categories are $p_ip_j$, $p_i(1-p_j)$, $(1-p_i)p_j$ and $(1-p_i)(1-p_j)$, respectively. Putting them together, $\boldsymbol N=(N_1,N_2,N_3,N_4)$ is distributed according to a multinomial distribution, $\textnormal{Multi}(m, p_ip_j, p_i(1-p_j), (1-p_i)p_j, (1-p_i)(1-p_j))$.

\begin{prop}\label{prop:expectation}If $\by_i$ and $\by_j$ are independent, then
\[\mathbb{E}(T(\by_i,\by_j)) =\frac{
%(1-(1-p_i-p_j+p_ip_j)^m)
p_i p_j}{p_i + p_j - p_i p_j}.\]
\end{prop}

\begin{proof} First, we compute conditional expectation given $N_1+N_2+N_3$. We observe that $N_1|N_1+N_2+N_3$ follows $\textnormal{Bernoulli}(N_1+N_2+N_3,\frac{p_ip_j}{p_i+p_j-p_ip_j})$. Hence, on set $N_1+N_2+N_3>0$, we have
\begin{align*} 
\mathbb{E}(T(\by_i,\by_j)|N_1+N_2+N_3) &= \mathbb{E}\left(\frac{N_1}{N_1+N_2+N_3}|N_1+N_2+N_3\right)\\
&= \frac{\mathbb{E}(N_1|N_1+N_2+N_3)}{N_1+N_2+N_3}\\
&= \frac{\frac{p_ip_j}{p_i+p_j-p_ip_j}(N_1+N_2+N_3)}{N_1+N_2+N_3}\\
&= \frac{p_ip_j}{p_i+p_j-p_ip_j}
\end{align*}
and on set $N_1+N_2+N_3=0$, we have
\[\mathbb{E}(T(\by_i,\by_j)|N_1+N_2+N_3)=\frac{p_ip_j}{p_i+p_j-p_ip_j}\]
Therefore,
\begin{align*}\mathbb{E}(T(\by_i,\by_j))&=\mathbb{E}[\mathbb{E}(T(\by_i,\by_j)|N_1+N_2+N_3)]\\
&= \frac{p_ip_j}{p_i+p_j-p_ip_j}\mathbb{P}(N_1+N_2+N_3=0) \\
& \hspace{2em} +\frac{p_ip_j}{p_i+p_j-p_ip_j}\mathbb{P}(N_1+N_2+N_3>0)\\&=
\frac{p_ip_j}{p_i+p_j-p_ip_j}.%(1-(1-p_i-p_j+p_ip_j)^m).
\end{align*}
\end{proof}

This allows us to define the centered Jaccard/Tanimoto coefficient as
\begin{equation} \label{def:centeredtanimoto}
T^{c}(\by_i,\by_j) = T(\by_i,\by_j) - \mathbb E[T(\by_i,\by_j)]
\end{equation}
This accounts for expected values, naturally distinguishing negative and positive associations. Generally, we would like to measure the deviation of an observed coefficient from an expected value, instead of simply looking at a magnitude of an observed statistics. Furthermore, this centered coefficient may be scaled by variance in order to span a pre-defined range.

To evaluate whether $\by_i$ and $\by_j$ are independent, a following statistical hypothesis testing is performed:
\begin{align} 
\label{def:hypothesis}
\begin{split}
H_{0} &: T^{c}(\by_i,\by_j) = 0 \\
H_{1} &: T^{c}(\by_i,\by_j) \neq 0.
\end{split}
\end{align}
The null hypothesis $H_{0}$ is that the centered Jaccard/Tanimoto coefficient equals zero. Note that this is equivalent to that the conventional (uncentered) Jaccard/Tanimoto coefficient equals an expected value under independence. Therefore, although we propose and use the centered coefficient, this hypothesis testing is attributed to both uncentered and centered versions. Then, a p-value indicates a probability of observing a coefficient equal to or more extreme than an observed coefficient under the null hypothesis.

\subsection*{Distribution of the Jaccard/Tanimoto Coefficient}
To obtain its p-value, we derive the distribution of Jaccard/Tanimoto coefficient under the null hypothesis. In terms of $\boldsymbol N=(N_1,N_2,N_3,N_4)$, the Jaccard/Tanimoto coefficient can be expressed as
\[
T(\by_i,\by_j)=\begin{cases}
\frac{N_1}{N_1+N_2+N_3} &\text{if } N_1+N_2+N_3>0\\
\frac{p_ip_j}{p_i+p_j-p_ip_j}&\text{otherwise.}
\end{cases}
\]

When $p_i$ and $p_j$ are known, the p-value is given by $\mathbb{P}(K_{T^c})$ where  \begin{equation}\label{def:crtitical region}
K_{T^c}=\left\{(N_1,N_2,N_3,N_4)\colon \left |\frac{N_1}{N_1+N_2+N_3} - \mathbb E [T(y_i,y_j)]\right|\geq |T^c|\right\}\;.
\end{equation}

However, in practice, probabilities $p_i$ and $p_j$ are usually unknown. Therefore, we define the centered Jaccard/Tanimoto coefficient by $\hat T^c = T- \frac{\hat p_i\hat p_j}{\hat p_i +\hat p_j- \hat p_i\hat p_j}$, where $\hat p_i =\frac{\sum{\by_i}}{m}$, $\hat p_j =\frac{\sum{\by_j}}{m}$ are standard estimators of $p_i$ and $p_j$ respectively. Plug-in estimates of $\mathbb E[T(y_i,y_j)]$ into Eq. \eqref{def:crtitical region} will result in conservative behaviors, since we estimate the probabilities on the same sample that we want to perform the test. Then, the estimates of expectation are biased toward the observed value of Jaccard/Tanimoto coefficient. To overcome this bias, we estimate probabilities $p_i$ and $p_j$ for each configuration $(N_1,N_2,N_3,N_4)$ separately. 

So in this case, the critical region is defined as follows
\begin{equation}
\label{def:crtitical_region_corrected}
 K_{\hat T^c}=\left\{(N_1,N_2,N_3,N_4)\colon \left |\frac{N_1}{N_1+N_2+N_3} -\frac{\tilde p_i\tilde p_j}{\tilde p_i+\tilde p_j -\tilde p_i\tilde p_j}\right|\geq |\hat T^c|\right\}\;,
\end{equation}
where $\tilde p_i = \frac{N_1+N_2}{m}$ and
$\tilde p_j = \frac{N_1+N_3}{m}$. 

Because the exact distribution is computationally expensive (see \emph{Results} for comparison), we introduce an asymptotic approximation when $m \to \infty$. It may be useful when dealing with very large binary data, where computational power is a bottleneck. Denote by $q_1=p_i p_j$ the probability that both $\by_i$ and $\by_j$ have ones, and by $q_2=p_i+p_j-2p_ip_j$ the probability that only one of two vectors has one. Similarly, $\hat q_1$ and $\hat q_2$ are defined with the plug-in estimators. As $m \to \infty$, we can estimate the variance:

\begin{prop}\label{prop:asymptotic}If $\by_i$ and $\by_j$ are independent then
\[
\sqrt{m}T^c(\by_i,\by_j)\to\mathcal{N}(0,\sigma^2)\]
as $m\to\infty$, where
\[
\sigma^2=\frac{q_1q_2(1-q_2)}{(q_1+q_2)^3}.
\]
\end{prop}

\begin{proof}
Theorem 14.6 of \cite{AllofStatistics} states that  
\[\sqrt{m}\left((N_1,N_2+N_3)/m-(q_1,q_2)\right)\to\mathcal{N}(0,\Sigma)\]
where
\[\Sigma=\begin{bmatrix}
q_1(1-q_1)&-q_1q_2\\
-q_1q_2&q_2(1-q_2)
\end{bmatrix}.
\]
Then, we define function $g(x_1,x_2)=\frac{x_1}{x_1+x_2}$ and apply the delta method. So, we get
\[
\sqrt{m}\left(T(\by_i,\by_j)-\frac{q_1}{q_1+q_2}\right)\to\mathcal{N}(0,\nabla g(q_1,q_2)\Sigma\nabla g(q_1,q_2)^T).\]
The gradient of $g$ is
\[\nabla g(x_1,x_2)=\left[\frac{x_2}{(x_1+x_2)^2},\frac{-x_1}{(x_1+x_2)^2}
\right].\]
Finally, after simplification, we obtain
\[\nabla g(q_1,q_2)\Sigma\nabla g(q_1,q_2)^T=\frac{q_1q_2(1-q_2)}{(q_1+q_2)^3}.\]
\end{proof}

In practice, probabilities $p_i$ and $p_j$ are unknown and need to be estimated. Recall that $\hat p_i=\frac{\# \{y_{ik}=1\}}{m}$ and $\hat p_j=\frac{\# \{y_{jk}=1\}}{m}$. We define $\hat q_1$ and $\hat q_2$ by replacing in definition of $q_1$ and $q_2$ true probabilities $p_i$ and $p_j$ by its estimators. So based on Proposition~\ref{prop:asymptotic} we are able to approximate p-values  as follow:
\begin{equation}
2\phi\left(\frac{\sqrt{m}}{\sigma}\left(T(\by_i,\by_j)-\frac{\hat q_1}{\hat q_1+\hat q_2}\right)\right)-1\;,
\end{equation}
where $\phi = \frac{1}{\sqrt{2\pi}} \int_{- \infty}^{x} e^{-x^2/2}dx$ is a standard Gaussian cumulative distribution function (CDF).

\subsection*{Measure Concentration Algorithm}

The distribution of the centered Jaccard/Tanimoto coefficient can be expressed in terms of the multinomial distribution. However, evaluating a significance test based on this representation requires exhaustive computations. It needs summation over all possible states of the multinomial distribution. For the centered Jaccard/Tanimoto coefficient between $y_i$ and $y_j$, we need to compute probability of event $ K_{\hat T^c}$ defined by Eq. \eqref{def:crtitical_region_corrected}. 

This can be quickly and accurately estimated by the measure concentration algorithm (MCA) with a known error bound \cite{Lacki2017}. For every $\varepsilon>0$, we will construct $I_\varepsilon$, a set of $(N_1,N_2,N_3,N_4)$ with $N_1+N_2+N_3+N_4=m$, such that $\mathbb P (N_1,N_2,N_3,N_4)\in I_\varepsilon \geq 1-\varepsilon$. Given the set $I_\varepsilon$, we have following bounds
\[p^L_\varepsilon(\hat T^c) = \mathbb P\left(K_{\hat T^c}\cap I_\varepsilon\right)\leq
\mathbb P\left(
K_{\hat T^c}\right)\leq
\mathbb P\left(
K_{\hat T^c}\cap I_\varepsilon\right)+\varepsilon=
p^U_\varepsilon(\hat T^c).
\]
In addition, $p^U_\varepsilon(\hat T^c)-p^L_\varepsilon(\hat T^c)=\varepsilon$.

The idea behind the algorithm is that a multinomial distribution concentrates around its mode. Two possible states $\boldsymbol N= (N_1,N_2,N_3,N_4)$ and  $\boldsymbol N^\prime=(N_1^\prime,N_2^\prime,N_3^\prime,N_4^\prime)$ are neighbors, $\boldsymbol{N}\sim\boldsymbol{N^\prime}$, if $\sum_{i=1}^4|N_i-N^\prime_i|=2$. This means that $\boldsymbol N^\prime$ can be obtained from $\boldsymbol N$ by moving one element to a different class. We construct the set $I_\varepsilon$ as follows. 

At the onset, $I_\varepsilon$ contains only the mode of multinomial distribution. We find the mode by a simple hill climbing algorithm, which starts with a state close to the mean of the multinomial distribution and follows the direction of increasing probability until the maximum is reached. Because of unimodality, it is indeed a global maximum. In the next steps, we add the neighbors of states which were previously visited. The procedure is repeated until the total probability of set $I_\varepsilon$ reaches the desired value $1-\varepsilon$. The details of the above method can be found in \cite{Lacki2017}. We construct the set $I_\varepsilon$ and we estimate the p-value by
\begin{align} 
 p^L(\hat T^c) = &
 \sum_{\boldsymbol N \in I_\varepsilon} \mathbf{1}\left(\left|\frac{N_1}{N_1+N_2+N_3} - \frac{\tilde p_i\tilde p_j}{\tilde p_i +\tilde p_j-\tilde p_i\tilde p_j}\right|\geq |\hat T^c|\right) \mathbb{P}(N_1,N_2,N_3,N_4).
\end{align} 

\subsection*{Bootstrap Procedure}
The bootstrap procedure has gained mainstream popularity for its wide applicability and statistical treatments \cite{Efron1994}. Creating an empirical distribution of null statistics allows for a flexible and robust estimation of p-values and related statistics. We show how to use the resampling with replacement to obtain statistical significance of $T^{c}(\by_i,\by_j)$. Particularly, resampling with replacement $\by_i$ and $\by_j$, separately, breaks any potential dependency. This allows us to calculate an empirical distribution of Jaccard/Tanimoto coefficients under the null hypothesis:

\

\begin{algorithm}[H]
  \KwIn{two binary vectors $\by_i$ and $\by_j$}
  \KwOut{$p$-$value$}
  
Calculate a centered Jaccard/Tanimoto coefficient $t = T^{c}(\by_i,\by_j)$. \\
\For{$b\leftarrow 1$ \KwTo $B$} {
Resample with replacement $\by_i$ and $\by_j$, resulting in $\by^*_i$ and $\by^*_j$. \\
Calculate bootstrap null coefficients $t^*_b = T^{c}(\by^*_ i,\by^*_j)$.
}
Compute the p-value by $$p\textnormal{-}value = \frac{\mathbf{1} \{ |t^*_b| \geq |t|; b=1, \ldots, B\}}{B}.$$

\caption{Bootstrap Procedure for Jaccard/Tanimoto Coefficients}\label{bootstrap}
\end{algorithm}

\

The expectation of Jaccard/Tanimoto coefficients is estimated directly from resampled vectors $\by^*_i$ and $\by^*_j$, that are effectively independent. Therefore, each iteration provides randomness, which helps avoid a bias related to using an estimated expectation based only on observation. Previously, there are early works in Monte Carlo procedures \cite{Connor1978,Gilpin1982} and published statistical tables for assessing randomness in species co-occurances \cite{Baroni-Urbani1976, Baroni-Urbani1979}. However, earlier works have assumed that a probability of occurrences is 0.5 regardless of species or biogeographic units. Permutation methods based on conventional uncentered coefficients are available in R packages, whose operating characteristics are not described in details \cite{Gotelli2015, Oksanen2017}.

The resolution of the empirical null distribution depends on $B$, where the larger $B$ will result in more precise estimation of p-values. Although the choice of $B$ would likely be dictated by $n$ and $m$, as well as available computational power, we recommend setting $B$ to at least 5-10 times of $m$. In our simulation studies, the total bootstrap iterations is set to $B = 5 \times m$, which are shown to be both accurate and fast. When comparing a very large set of species or OTUs, it may be helpful to pool null statistics to increase the p-value resolution and speed up the computation.

\section*{Results and discussion}
\subsection*{Simulation Studies}

We have developed statistical methods and algorithms to obtain statistical significance of Jaccard/Tanimoto similarity coefficients for biological presence-absence data. Beyond deriving the exact solution, we introduce the measurement concentration algorithm (MCA) and bootstrap method. We characterize their operating characteristics by comprehensive simulation studies where a wide range of parameters for presence-absence datasets are considered. Our goal is to maintain theoretically correct behaviors of p-values. Null p-values corresponding to $H_0$ are evaluated against a Uniform(0,1) distribution. False discovery rates (FDRs) are directly estimated from p-values produced by our methods to demonstrate an overall error control.

First, we conducted 5 simulation scenarios using different underlying occurrence probabilities $p= 0.1, 0.3, 0.5, 0.7, 0.9$ to generate independent presence-absence datasets. In essence, they are two species of length $m=100$ that exhibit unrelated co-occurrence patterns, where a proportion of presence (1's) ranges from $10\%$ to $90\%$. For each of simulation scenarios, a total of $2000$ comparisons were made using a length $m=100$. Without any information about simulation parameters, our proposed methods are applied on an identically simulated dataset (Figure 2). Theoretically correct p-values under the null hypothesis (null p-values) should form a Uniform distribution between 0 and 1, which are denoted by dashed diagonal lines in QQ plots. An upward deviation from diagonals shows an anti-conservative bias, as shown among some asymptotic p-values. In all scenarios, p-values from the exact solution, bootstrap ($B=500$), and measure concentration (accuracy $=1 \times 10^{-5}$) algorithms follow a theoretically correct Uniform(0,1) distribution (Figure 2). Asymptotic approximation is inconsistent; its behavior is anti-conservative with $p=0.3, 0.5$ and slightly conservative with $p=0.7, 0.9$. Asymptotic approximation should only be used when computational time is a critical bottleneck.

Second, we generated a mixture of independent and dependent datasets out of $n=2000$ presence-absence vectors (of $n=2000$ species observed in $m=200$ biogeographic units) to evaluate false discovery rates. In three separate scenarios, we simulated $25\%$, $50\%$, and $75\%$ of $n=2000$ species to be independent, resulting in null proportions of $\pi_0 = .25, .50, .75$ respectively. For example, a scenario with $\pi_0 = .75$ produces 500 out of $n=2000$ presence-absence variables that are truly associated with the query variable. Then, our proposed asymptotic approximation, bootstrap method, and measure concentration algorithm (MCA) are used to automatically compute p-values. To account for variation in simulation, we repeated each scenario 20 times. False discovery rates (FDRs) and $\pi_0$ are estimated by the q-value methodology \cite{storey2003statistical}. Q-values are evaluated against FDR thresholds, so that we can evaluate accuracy of observed FDRs (Figure 3). Twenty simulation replications are shown in semi-transparent shades, whereas their group averages for 3 methods are shown as solid lines. An upward deviation as shown by asymptotic approximation indicates an overall anti-conservative behavior, likely due to $m \not\to \infty$. The bootstrap and MCA maintain the overall error rates, where the bootstrap exhibits slightly conservative characteristics (Figure 3).

Third, we compared the computational efficiency of our proposed methods using our {\tt jaccard} package on RStudio Cloud (Intel Xeon 2.90GHz and 1GB RAM), with R 3.5.0. We measured the runtime for a range of lengths $m= 50, \ldots, 500$. For each $m$, we applied the proposed methods 10 times, with the bootstrap iteration $B=5 \times m$ and MCA accuracy of $1 \times 10^{-5}$. The average runtimes are shown Figure 4. Our proposed computational methods show drastic improvement over the exact solution as $m$ increases. The asymptotic approximation is mostly instantaneous. When the similarity between two presence-absence vectors of length $m=500$ were tested using the {\tt jaccard} package, the exact solution was prohibitively slow, taking 41.5s on average. The bootstrap method was 449.8 times (0.09s) faster, whereas MCA was 92.5 times (0.45s) faster than the exact solution. Furthermore, we compared the runtimes of estimation methods for $m = 1000, \ldots, 10000$ (Figure S1). The gain in computational efficiency is more pronounced as the dimension (i.e., a length of presence-absence vectors) grows in size.

Last, a simulation study with $p= 0.5$ and $m=200$ was used to evaluate two recent methods of species co-occurrences analysis. We generated independent presence-absence data where two species are truly unrelated. Then, methods of combinatorics \cite{Veech2013} and hypergeometric distributions \cite{Griffith2016} are applied to obtain p-values. We followed the recommendations given in each paper, displaying four possible p-values from \cite{Veech2013} (Figure S2) and two one-sided p-values from \cite{Griffith2016} (Figure S3). We observe these p-values under the null hypothesis to substantially deviate from theoretically correct Uniform(0,1) distributions.

\subsection*{Applications in Species Co-occurrences}
To show applications in statistically testing biological presence-absence data, the proposed methods are applied to species co-occurrence data. We investigated bird species on 28 islands in the Republic of Vanuatu, that are available in \cite{Manly2006} and analyzed in several pioneering studies in non-random co-occurrences of species \cite{Connor1979, Gilpin1982, Wilson1987, Manly1995}. The data is consisted of presence and absence of bird species in 28 islands of Vanuatu, which used to be known as the New Hebrides. Three generalist species that existed in all 28 islands were removed from our analysis. We are interested in identifying what pairs of species exhibit statistically significantly co-occurrences. 

For $n=53$ bird species in $m=28$ islands, we obtained 1378 pair-wise Jaccard/Tanimoto similarity coefficients. The conventional Jaccard/Tanimoto coefficients depends strongly on their expected values under independence (Figure 5). Similarly, the conventional Jaccard/Tanimoto coefficients are substantially correlated with the proportion of occurrences, with a Pearson correlation of $0.43$ (p-value $< 2.2 \times 10^{-16}$). Relying only on similarity coefficients would miss non-random co-occurrences among bird species that live in a few islands (Figure S4). Our proposed methods account for co-occurrences that would be expected under independence. Histograms of the uncentered and centered Jaccard/Tanimoto coefficients are compared in Figure S5.

We computed statistical significance by applying the bootstrap method with $B=5000$ and MCA with accuracy of $1 \times 10^{-5}$. Our two computational approaches estimated p-values that are almost identical with their mean squared deviation of $1.15 \times 10^{-4}$ (Figure S6). Significant results that are substantially deviating from random samples indicate non-random co-occurrences of species (Figure 6). Out of 1378 pairs of species that were tested, the proportion of independent specie pairs was estimated to be $24\%$ using q-value methodology \cite{storey2003statistical}. Then, we calculated false discovery rates (FDRs) from 1378 pair-wise p-values. We discovered that 374 ($27\%$) pairs are deemed significant at a q-value threshold of 0.10.
%By definition, when these 374 pairs are declared to be dependent (i.e., significantly co-occurring), at most $10\%$ of these pairs may be false discoveries.

Additionally, we applied the Jaccard/Tanimoto similarity tests among fish species in French freshwater streams, surveyed over a long period of time \cite{Comte2016}. Briefly, the presence and absence data of the $n=32$ most common fish species in $m=3347$ sites across French rivers are obtained during 1980 - 1991 \cite{Comte2016}. Our analysis estimates that about $84.3\%$ of 496 pairs are estimated to be non-randomly co-occurring. As surveyed for over a decade across Fresh rivers and surrounding habitats, it is reasonable that many fish species are interacting or influenced by related climate conditions. There are 21 pairs of species with q-values $> 0.1$ (corresponding p-values ranging from 0.637 and 0.969). For example, the centered statistics between \emph{Pungitius pungitius} and \emph{Cyprinus carpio} is $3.31 \times 10^{-4}$, whereas that between \emph{Pungitius pungitius} and \emph{Lota lota} is $-4.40 \times 10^{-4}$. \emph{P. pungitius} is a small fish species typically riding in thick submerged vegetation with the breeding season falling in April - July. \emph{C. carpio} and \emph{L. lota} are much bigger species and generally prefers a large body of water.

\section*{Conclusion}
From biogeography to microbiology, evaluating similarity among species and biogeographic units is fundamental to assessing co-existence and biodiversity. Having observed occurrences of species in multiple biogeographic units, one of the primary goals in analyzing presence-absence data is to identify non-random co-occurrences. Even if two species would be present independently of each other, they may occur together by chance. For the last 30 years, the Jaccard/Tanimoto coefficient has been shown to be highly useful for quantitative analysis of co-occurrences that help inform systematic relationship among species \cite{Birks1987, Jackson1992, RealVargas1996}. We have developed a rigorous statistical framework and methods to efficiently calculate statistical significance of such similarity and to identify non-random co-occurrences.

For testing co-occurrences using the Jaccard/Tanimoto coefficient, we introduce exact and asymptotic solutions, as well as bootstrap and measure concentration algorithm. The proposed suite of statistical methods can provide a rigorous guideline to identify related species. Through comprehensive simulation studies, we characterized their operating characteristics using p-values and FDRs. The proposed bootstrap and measure concentration algorithms are highly accurate and efficient, providing orders of magnitude improvement in a computational speed. We have implemented the proposed methods in an open source R package and a Shiny web app. A user can upload a dataset to be analyzed, and create histograms and heat maps automatically. This will facilitate adaptation of p-values, FDRs, and related quantities in analyzing species co-occurrences.

Beyond species co-occurrences, the Jaccard/Tanimoto coefficient is used in diverse areas of biological science where binary data are observed and compared. When molecules and reactions are represented as hashed fingerprints, it is used for quantitative comparisons and classifications \cite{Todeschini2012,Rahman2014,Bajusz2015}. Similarity between biochemical reactions can be tested by applying our methods on their corresponding fingerprints. In genomics, the standard tools such as BEDTools \cite{Quinlan2014} evaluate genomic intervals using the Jaccard/Tanimoto coefficients. Given genomic intervals from two samples or groups, one could test whether their overlap is statistically significant, providing evidences for shared genomic variations. Due to the popularity of Jaccard/Tanimoto coefficients, the proposed suite of methods would be useful in a broad range of scientific applications.

\begin{backmatter}
\bibliographystyle{bmc-mathphys} 
\bibliography{refs}   

%% BioMed_Central_Bib_Style_v1.01

\begin{thebibliography}{38}
% BibTex style file: bmc-mathphys.bst (version 2.1), 2014-07-24
\ifx \bisbn   \undefined \def \bisbn  #1{ISBN #1}\fi
\ifx \binits  \undefined \def \binits#1{#1}\fi
\ifx \bauthor  \undefined \def \bauthor#1{#1}\fi
\ifx \batitle  \undefined \def \batitle#1{#1}\fi
\ifx \bjtitle  \undefined \def \bjtitle#1{#1}\fi
\ifx \bvolume  \undefined \def \bvolume#1{\textbf{#1}}\fi
\ifx \byear  \undefined \def \byear#1{#1}\fi
\ifx \bissue  \undefined \def \bissue#1{#1}\fi
\ifx \bfpage  \undefined \def \bfpage#1{#1}\fi
\ifx \blpage  \undefined \def \blpage #1{#1}\fi
\ifx \burl  \undefined \def \burl#1{\textsf{#1}}\fi
\ifx \doiurl  \undefined \def \doiurl#1{\textsf{#1}}\fi
\ifx \betal  \undefined \def \betal{\textit{et al.}}\fi
\ifx \binstitute  \undefined \def \binstitute#1{#1}\fi
\ifx \binstitutionaled  \undefined \def \binstitutionaled#1{#1}\fi
\ifx \bctitle  \undefined \def \bctitle#1{#1}\fi
\ifx \beditor  \undefined \def \beditor#1{#1}\fi
\ifx \bpublisher  \undefined \def \bpublisher#1{#1}\fi
\ifx \bbtitle  \undefined \def \bbtitle#1{#1}\fi
\ifx \bedition  \undefined \def \bedition#1{#1}\fi
\ifx \bseriesno  \undefined \def \bseriesno#1{#1}\fi
\ifx \blocation  \undefined \def \blocation#1{#1}\fi
\ifx \bsertitle  \undefined \def \bsertitle#1{#1}\fi
\ifx \bsnm \undefined \def \bsnm#1{#1}\fi
\ifx \bsuffix \undefined \def \bsuffix#1{#1}\fi
\ifx \bparticle \undefined \def \bparticle#1{#1}\fi
\ifx \barticle \undefined \def \barticle#1{#1}\fi
\ifx \bconfdate \undefined \def \bconfdate #1{#1}\fi
\ifx \botherref \undefined \def \botherref #1{#1}\fi
\ifx \url \undefined \def \url#1{\textsf{#1}}\fi
\ifx \bchapter \undefined \def \bchapter#1{#1}\fi
\ifx \bbook \undefined \def \bbook#1{#1}\fi
\ifx \bcomment \undefined \def \bcomment#1{#1}\fi
\ifx \oauthor \undefined \def \oauthor#1{#1}\fi
\ifx \citeauthoryear \undefined \def \citeauthoryear#1{#1}\fi
\ifx \endbibitem  \undefined \def \endbibitem {}\fi
\ifx \bconflocation  \undefined \def \bconflocation#1{#1}\fi
\ifx \arxivurl  \undefined \def \arxivurl#1{\textsf{#1}}\fi
\csname PreBibitemsHook\endcsname

%%% 1
\bibitem{Jaccard1912}
\begin{barticle}
\bauthor{\bsnm{Jaccard}, \binits{P.}}:
\batitle{The distribution of the flora in the alpine zone}.
\bjtitle{New Phytologist}
\bvolume{11}(\bissue{2}),
\bfpage{37}--\blpage{50}
(\byear{1912}).
doi:\doiurl{10.1111/j.1469-8137.1912.tb05611.x}
\end{barticle}
\endbibitem

%%% 2
\bibitem{Tanimoto1958}
\begin{botherref}
\oauthor{\bsnm{Tanimoto}, \binits{T.}}:
An elementary mathematical theory of classification and prediction.
Technical report,
International Business Machines Corporation,
New York
(1958)
\end{botherref}
\endbibitem

%%% 3
\bibitem{Birks1987}
\begin{barticle}
\bauthor{\bsnm{Birks}, \binits{H.J.B.}}:
\batitle{Recent methodological developments in quantitative descriptive
  biogeography}.
\bjtitle{Ann. Zool. Fenn.}
\bvolume{24},
\bfpage{165}--\blpage{178}
(\byear{1987})
\end{barticle}
\endbibitem

%%% 4
\bibitem{Jackson1992}
\begin{barticle}
\bauthor{\bsnm{Jackson}, \binits{D.A.}},
\bauthor{\bsnm{Somers}, \binits{K.M.}},
\bauthor{\bsnm{Harvey}, \binits{H.H.}}:
\batitle{Null models and fish communities: Evidence of nonrandom patterns}.
\bjtitle{The American Naturalist}
\bvolume{139}(\bissue{5}),
\bfpage{930}--\blpage{951}
(\byear{1992})
\end{barticle}
\endbibitem

%%% 5
\bibitem{RealVargas1996}
\begin{barticle}
\bauthor{\bsnm{Real}, \binits{R.}},
\bauthor{\bsnm{Vargas}, \binits{J.M.}}:
\batitle{The probabilistic basis of jaccard's index of similarity}.
\bjtitle{Systematic Biology}
\bvolume{45}(\bissue{3}),
\bfpage{380}--\blpage{385}
(\byear{1996}).
doi:\doiurl{10.1093/sysbio/45.3.380}
\end{barticle}
\endbibitem

%%% 6
\bibitem{Manly2006}
\begin{bbook}
\bauthor{\bsnm{Manly}, \binits{B.F.J.}}:
\bbtitle{Randomization, Bootstrap and Monte Carlo Methods in Biology}.
\bpublisher{Chapman \& Hall / CRC Press},
\blocation{Boca Raton, FL}
(\byear{2006})
\end{bbook}
\endbibitem

%%% 7
\bibitem{Davies1993}
\begin{bbook}
\bauthor{\bsnm{Davies}, \binits{N.B.}},
\bauthor{\bsnm{Krebs}, \binits{J.R.}}:
\bbtitle{An Introduction to Behavioural Ecology}.
\bpublisher{Wiley-Blackwell},
\blocation{U.S.A.}
(\byear{1993})
\end{bbook}
\endbibitem

%%% 8
\bibitem{Townsend2002}
\begin{bbook}
\bauthor{\bsnm{Townsend}, \binits{C.R.}},
\bauthor{\bsnm{Begon}, \binits{M.}},
\bauthor{\bsnm{Harper}, \binits{J.L.}}:
\bbtitle{Essentials of Ecology}.
\bpublisher{Wiley-Blackwell},
\blocation{U.S.A.}
(\byear{2002})
\end{bbook}
\endbibitem

%%% 9
\bibitem{Whittaker1960}
\begin{barticle}
\bauthor{\bsnm{Whittaker}, \binits{R.H.}}:
\batitle{Vegetation of the siskiyou mountains, oregon and california}.
\bjtitle{Ecological Monographs}
\bvolume{30}(\bissue{3}),
\bfpage{279}--\blpage{338}
(\byear{1960}).
doi:\doiurl{10.2307/1943563}
\end{barticle}
\endbibitem

%%% 10
\bibitem{Harrison1992}
\begin{barticle}
\bauthor{\bsnm{Harrison}, \binits{S.}},
\bauthor{\bsnm{Ross}, \binits{S.J.}},
\bauthor{\bsnm{Lawton}, \binits{J.H.}}:
\batitle{Beta diversity on geographic gradients in britain}.
\bjtitle{The Journal of Animal Ecology}
\bvolume{61}(\bissue{1}),
\bfpage{151}
(\byear{1992}).
doi:\doiurl{10.2307/5518}
\end{barticle}
\endbibitem

%%% 11
\bibitem{Koleff2003}
\begin{barticle}
\bauthor{\bsnm{Koleff}, \binits{P.}},
\bauthor{\bsnm{Gaston}, \binits{K.J.}},
\bauthor{\bsnm{Lennon}, \binits{J.J.}}:
\batitle{Measuring beta diversity for presence-absence data}.
\bjtitle{Journal of Animal Ecology}
\bvolume{72}(\bissue{3}),
\bfpage{367}--\blpage{382}
(\byear{2003}).
doi:\doiurl{10.1046/j.1365-2656.2003.00710.x}
\end{barticle}
\endbibitem

%%% 12
\bibitem{Connor1979}
\begin{barticle}
\bauthor{\bsnm{Connor}, \binits{E.F.}},
\bauthor{\bsnm{Simberloff}, \binits{D.}}:
\batitle{The assembly of species communities: Chance or competition?}
\bjtitle{Ecology}
\bvolume{60}(\bissue{6}),
\bfpage{1132}
(\byear{1979}).
doi:\doiurl{10.2307/1936961}
\end{barticle}
\endbibitem

%%% 13
\bibitem{Diamond1982}
\begin{barticle}
\bauthor{\bsnm{Diamond}, \binits{J.M.}},
\bauthor{\bsnm{Gilpin}, \binits{M.E.}}:
\batitle{Examination of the "null" model of connor and simberloff for species
  co-occurrence on islands}.
\bjtitle{Oecologia}
\bvolume{52},
\bfpage{64}--\blpage{74}
(\byear{1982}).
doi:\doiurl{10.1007/BF00349013}
\end{barticle}
\endbibitem

%%% 14
\bibitem{Gilpin1982}
\begin{barticle}
\bauthor{\bsnm{Gilpin}, \binits{M.E.}},
\bauthor{\bsnm{Diamond}, \binits{J.M.}}:
\batitle{Factors contributing to non-randomness in species co-occurrences on
  islands}.
\bjtitle{Oecologia}
\bvolume{52},
\bfpage{75}--\blpage{84}
(\byear{1982}).
doi:\doiurl{10.1007/BF00349014}
\end{barticle}
\endbibitem

%%% 15
\bibitem{Wilson1987}
\begin{barticle}
\bauthor{\bsnm{Wilson}, \binits{J.B.}}:
\batitle{Methods for detecting non-randomness in species co-occurrences: a
  contribution}.
\bjtitle{Oecologia}
\bvolume{73}(\bissue{4}),
\bfpage{579}--\blpage{582}
(\byear{1987}).
doi:\doiurl{10.1007/BF00379419}
\end{barticle}
\endbibitem

%%% 16
\bibitem{Manly1995}
\begin{barticle}
\bauthor{\bsnm{Manly}, \binits{B.F.J.}}:
\batitle{A note on the analysis of species co-occurrences}.
\bjtitle{Ecology}
\bvolume{76}(\bissue{4}),
\bfpage{1109}--\blpage{1115}
(\byear{1995}).
doi:\doiurl{10.2307/1940919}
\end{barticle}
\endbibitem

%%% 17
\bibitem{Sanderson1998}
\begin{barticle}
\bauthor{\bsnm{Sanderson}, \binits{J.}},
\bauthor{\bsnm{Moulton}, \binits{M.}},
\bauthor{\bsnm{Selfridge}, \binits{R.}}:
\batitle{Null matrices and the analysis of species co-occurrencessanderson}.
\bjtitle{Oecologia}
\bvolume{116}(\bissue{1--2}),
\bfpage{275}--\blpage{283}
(\byear{1998}).
doi:\doiurl{10.1007/s004420050}
\end{barticle}
\endbibitem

%%% 18
\bibitem{Ellwood2009}
\begin{barticle}
\bauthor{\bsnm{Ellwood}, \binits{M.D.F.}},
\bauthor{\bsnm{Manica}, \binits{A.}},
\bauthor{\bsnm{Foster}, \binits{W.A.}}:
\batitle{Stochastic and deterministic processes jointly structure tropical
  arthropod communities}.
\bjtitle{Ecology Letters}
\bvolume{12}(\bissue{4}),
\bfpage{277}--\blpage{284}
(\byear{2009}).
doi:\doiurl{10.1111/j.1461-0248.2009.01284.x}
\end{barticle}
\endbibitem

%%% 19
\bibitem{Chase2011}
\begin{barticle}
\bauthor{\bsnm{Chase}, \binits{J.M.}},
\bauthor{\bsnm{Myers}, \binits{J.A.}}:
\batitle{Disentangling the importance of ecological niches from stochastic
  processes across scales}.
\bjtitle{Philosophical transactions of the Royal Society B: Biological
  sciences}
\bvolume{366}(\bissue{1576}),
\bfpage{2351}--\blpage{2363}
(\byear{2011}).
doi:\doiurl{10.1098/rstb.2011.0063}
\end{barticle}
\endbibitem

%%% 20
\bibitem{Fridley2007}
\begin{barticle}
\bauthor{\bsnm{Fridley}, \binits{J.D.}},
\bauthor{\bsnm{Vandermast}, \binits{D.B.}},
\bauthor{\bsnm{Kuppinger}, \binits{D.M.}},
\bauthor{\bsnm{Manthey}, \binits{M.}},
\bauthor{\bsnm{Peet}, \binits{R.K.}}:
\batitle{Co-occurrence based assessment of habitat generalists and specialists:
  A new approach for the measurement of niche width}.
\bjtitle{Journal of Ecology}
\bvolume{95}(\bissue{4}),
\bfpage{707}--\blpage{722}
(\byear{2007}).
doi:\doiurl{10.1111/j.1365-2745.2007.01236.x}
\end{barticle}
\endbibitem

%%% 21
\bibitem{Araujo2013}
\begin{barticle}
\bauthor{\bsnm{Ara{\'{u}}jo}, \binits{M.B.}},
\bauthor{\bsnm{Rozenfeld}, \binits{A.}}:
\batitle{The geographic scaling of biotic interactions}.
\bjtitle{Ecography}
(\byear{2013}).
doi:\doiurl{10.1111/j.1600-0587.2013.00643.x}
\end{barticle}
\endbibitem

%%% 22
\bibitem{Baroni-Urbani1976}
\begin{barticle}
\bauthor{\bsnm{Baroni-Urbani}, \binits{C.}},
\bauthor{\bsnm{Buser}, \binits{M.W.}}:
\batitle{Similarity of binary data}.
\bjtitle{Systematic Zoology}
\bvolume{25}(\bissue{3}),
\bfpage{251}
(\byear{1976}).
doi:\doiurl{10.2307/2412493}
\end{barticle}
\endbibitem

%%% 23
\bibitem{Baroni-Urbani1979}
\begin{barticle}
\bauthor{\bsnm{Baroni-Urbani}, \binits{C.}}:
\batitle{A statistical table for the degree of coexistence between two
  species}.
\bjtitle{Oecologia}
\bvolume{44}(\bissue{3}),
\bfpage{287}--\blpage{289}
(\byear{1979}).
doi:\doiurl{10.1007/bf00545229}
\end{barticle}
\endbibitem

%%% 24
\bibitem{Veech2013}
\begin{barticle}
\bauthor{\bsnm{Veech}, \binits{J.A.}}:
\batitle{A probabilistic model for analysing species co-occurrence}.
\bjtitle{Global Ecology and Biogeography}
\bvolume{22},
\bfpage{252}--\blpage{260}
(\byear{2013}).
doi:\doiurl{10.1111/j.1466-8238.2012.00789.x}
\end{barticle}
\endbibitem

%%% 25
\bibitem{Griffith2016}
\begin{botherref}
\oauthor{\bsnm{Griffith}, \binits{D.M.}},
\oauthor{\bsnm{Veech}, \binits{J.A.}},
\oauthor{\bsnm{Marsh}, \binits{C.J.}}:
cooccur: Probabilistic species co-occurrence analysis inr.
Journal of Statistical Software
\textbf{69}
(2016).
doi:\doiurl{10.18637/jss.v069.c02}
\end{botherref}
\endbibitem

%%% 26
\bibitem{R}
\begin{bbook}
\bauthor{\bsnm{{R Core Team}}}:
\bbtitle{R: A Language and Environment for Statistical Computing}.
\bpublisher{R Foundation for Statistical Computing},
\blocation{Vienna, Austria}
(\byear{2017}).
\bcomment{R Foundation for Statistical Computing}.
\burl{https://www.R-project.org}
\end{bbook}
\endbibitem

%%% 27
\bibitem{AllofStatistics}
\begin{bbook}
\bauthor{\bsnm{Wasserman}, \binits{L.}}:
\bbtitle{All of Statistics: A Concise Course in Statistical Inference}.
\bpublisher{Springer},
\blocation{New York, U.S.A.}
(\byear{2010})
\end{bbook}
\endbibitem

%%% 28
\bibitem{Lacki2017}
\begin{barticle}
\bauthor{\bsnm{{\L}{\k{a}}cki}, \binits{M.K.}},
\bauthor{\bsnm{Startek}, \binits{M.}},
\bauthor{\bsnm{Valkenborg}, \binits{D.}},
\bauthor{\bsnm{Gambin}, \binits{A.}}:
\batitle{{IsoSpec}: Hyperfast fine structure calculator}.
\bjtitle{Analytical Chemistry}
\bvolume{89}(\bissue{6}),
\bfpage{3272}--\blpage{3277}
(\byear{2017}).
doi:\doiurl{10.1021/acs.analchem.6b01459}
\end{barticle}
\endbibitem

%%% 29
\bibitem{Efron1994}
\begin{bbook}
\bauthor{\bsnm{Efron}, \binits{B.}},
\bauthor{\bsnm{Tibshirani}, \binits{R.}}:
\bbtitle{An Introduction to the Bootstrap}.
\bpublisher{Chapman \& Hall / CRC Press},
\blocation{Boca Raton, Florida}
(\byear{1994})
\end{bbook}
\endbibitem

%%% 30
\bibitem{Connor1978}
\begin{barticle}
\bauthor{\bsnm{Connor}, \binits{E.F.}},
\bauthor{\bsnm{Simberloff}, \binits{D.}}:
\batitle{Species number and compositional similarity of the galapagos flora and
  avifauna}.
\bjtitle{Ecological Monographs}
\bvolume{48},
\bfpage{219}--\blpage{248}
(\byear{1978}).
doi:\doiurl{10.2307/2937300}
\end{barticle}
\endbibitem

%%% 31
\bibitem{Gotelli2015}
\begin{botherref}
\oauthor{\bsnm{Gotelli}, \binits{N.J.}},
\oauthor{\bsnm{Hart}, \binits{E.M.}},
\oauthor{\bsnm{Ellison}, \binits{A.M.}}:
EcoSimR: Null Model Analysis for Ecological Data.
(2015).
R package version 0.1.0.
\url{http://github.com/gotellilab/EcoSimR}
\end{botherref}
\endbibitem

%%% 32
\bibitem{Oksanen2017}
\begin{botherref}
\oauthor{\bsnm{Oksanen}, \binits{J.}},
\oauthor{\bsnm{Blanchet}, \binits{F.G.}},
\oauthor{\bsnm{Friendly}, \binits{M.}},
\oauthor{\bsnm{Kindt}, \binits{R.}},
\oauthor{\bsnm{Legendre}, \binits{P.}},
\oauthor{\bsnm{McGlinn}, \binits{D.}},
\oauthor{\bsnm{Minchin}, \binits{P.R.}},
\oauthor{\bsnm{O'Hara}, \binits{R.B.}},
\oauthor{\bsnm{Simpson}, \binits{G.L.}},
\oauthor{\bsnm{Solymos}, \binits{P.}},
\oauthor{\bsnm{Stevens}, \binits{M.H.H.}},
\oauthor{\bsnm{Szoecs}, \binits{E.}},
\oauthor{\bsnm{Wagner}, \binits{H.}}:
Vegan: Community Ecology Package.
(2017).
R package version 2.4-5.
\url{https://CRAN.R-project.org/package=vegan}
\end{botherref}
\endbibitem

%%% 33
\bibitem{storey2003statistical}
\begin{barticle}
\bauthor{\bsnm{Storey}, \binits{J.D.}},
\bauthor{\bsnm{Tibshirani}, \binits{R.}}:
\batitle{Statistical significance for genomewide studies}.
\bjtitle{Proceedings of the National Academy of Sciences}
\bvolume{100}(\bissue{16}),
\bfpage{9440}--\blpage{9445}
(\byear{2003}).
doi:\doiurl{10.1073/pnas.1530509100}
\end{barticle}
\endbibitem

%%% 34
\bibitem{Comte2016}
\begin{barticle}
\bauthor{\bsnm{Comte}, \binits{L.}},
\bauthor{\bsnm{Hugueny}, \binits{B.}},
\bauthor{\bsnm{Grenouillet}, \binits{G.}}:
\batitle{Climate interacts with anthropogenic drivers to determine extirpation
  dynamics}.
\bjtitle{Ecography}
\bvolume{39}(\bissue{10}),
\bfpage{1008}--\blpage{1016}
(\byear{2016}).
doi:\doiurl{10.1111/ecog.01871}
\end{barticle}
\endbibitem

%%% 35
\bibitem{Todeschini2012}
\begin{barticle}
\bauthor{\bsnm{Todeschini}, \binits{R.}},
\bauthor{\bsnm{Consonni}, \binits{V.}},
\bauthor{\bsnm{Xiang}, \binits{H.}},
\bauthor{\bsnm{Holliday}, \binits{J.}},
\bauthor{\bsnm{Buscema}, \binits{M.}},
\bauthor{\bsnm{Willett}, \binits{P.}}:
\batitle{Similarity coefficients for binary chemoinformatics data: Overview and
  extended comparison using simulated and real data sets}.
\bjtitle{J. Chem. Inf. Model.}
\bvolume{52}(\bissue{11}),
\bfpage{2884}--\blpage{2901}
(\byear{2012}).
doi:\doiurl{10.1021/ci300261r}
\end{barticle}
\endbibitem

%%% 36
\bibitem{Rahman2014}
\begin{barticle}
\bauthor{\bsnm{Rahman}, \binits{S.A.}},
\bauthor{\bsnm{Cuesta}, \binits{S.M.}},
\bauthor{\bsnm{Furnham}, \binits{N.}},
\bauthor{\bsnm{Holliday}, \binits{G.L.}},
\bauthor{\bsnm{Thornton}, \binits{J.M.}}:
\batitle{{EC}-{BLAST}: a tool to automatically search and compare enzyme
  reactions}.
\bjtitle{Nature Methods}
\bvolume{11}(\bissue{2}),
\bfpage{171}--\blpage{174}
(\byear{2014}).
doi:\doiurl{10.1038/nmeth.2803}
\end{barticle}
\endbibitem

%%% 37
\bibitem{Bajusz2015}
\begin{botherref}
\oauthor{\bsnm{Bajusz}, \binits{D.}},
\oauthor{\bsnm{R{\'{a}}cz}, \binits{A.}},
\oauthor{\bsnm{H{\'{e}}berger}, \binits{K.}}:
Why is tanimoto index an appropriate choice for fingerprint-based similarity
  calculations?
J Cheminform
\textbf{7}(1)
(2015).
doi:\doiurl{10.1186/s13321-015-0069-3}
\end{botherref}
\endbibitem

%%% 38
\bibitem{Quinlan2014}
\begin{botherref}
\oauthor{\bsnm{Quinlan}, \binits{A.R.}}:
Bedtools: the swiss-army tool for genome feature analysis.
Current Protocols in Bioinformatics,
11--12
(2014).
doi:\doiurl{10.1002/0471250953.bi1112s47}
\end{botherref}
\endbibitem

\end{thebibliography}

\newcommand{\BMCxmlcomment}[1]{}

\BMCxmlcomment{

<refgrp>

<bibl id="B1">
  <title><p>The distribution of the flora in the alpine zone</p></title>
  <aug>
    <au><snm>Jaccard</snm><fnm>P.</fnm></au>
  </aug>
  <source>New Phytologist</source>
  <publisher>Wiley-Blackwell</publisher>
  <pubdate>1912</pubdate>
  <volume>11</volume>
  <issue>2</issue>
  <fpage>37</fpage>
  <lpage>50</lpage>
  <url>http://dx.doi.org/10.1111/j.1469-8137.1912.tb05611.x</url>
</bibl>

<bibl id="B2">
  <title><p>An Elementary Mathematical theory of Classification and
  Prediction</p></title>
  <aug>
    <au><snm>Tanimoto</snm><fnm>T.</fnm></au>
  </aug>
  <publisher>New York</publisher>
  <pubdate>1958</pubdate>
</bibl>

<bibl id="B3">
  <title><p>Recent methodological developments in quantitative descriptive
  biogeography</p></title>
  <aug>
    <au><snm>Birks</snm><fnm>H. J. B.</fnm></au>
  </aug>
  <source>Ann. Zool. Fenn.</source>
  <pubdate>1987</pubdate>
  <volume>24</volume>
  <fpage>165</fpage>
  <lpage>178</lpage>
  <url>https://www.jstor.org/stable/23734493</url>
</bibl>

<bibl id="B4">
  <title><p>Null models and fish communities: Evidence of nonrandom
  patterns</p></title>
  <aug>
    <au><snm>Jackson</snm><fnm>D. A.</fnm></au>
    <au><snm>Somers</snm><fnm>K. M.</fnm></au>
    <au><snm>Harvey</snm><fnm>H. H.</fnm></au>
  </aug>
  <source>The American Naturalist</source>
  <pubdate>1992</pubdate>
  <volume>139</volume>
  <issue>5</issue>
  <fpage>930</fpage>
  <lpage>951</lpage>
</bibl>

<bibl id="B5">
  <title><p>The Probabilistic Basis of Jaccard's Index of
  Similarity</p></title>
  <aug>
    <au><snm>Real</snm><fnm>R.</fnm></au>
    <au><snm>Vargas</snm><fnm>J. M.</fnm></au>
  </aug>
  <source>Systematic Biology</source>
  <publisher>Oxford University Press ({OUP})</publisher>
  <pubdate>1996</pubdate>
  <volume>45</volume>
  <issue>3</issue>
  <fpage>380</fpage>
  <lpage>-385</lpage>
  <url>http://dx.doi.org/10.1093/sysbio/45.3.380</url>
</bibl>

<bibl id="B6">
  <title><p>Randomization, Bootstrap and Monte Carlo Methods in
  Biology</p></title>
  <aug>
    <au><snm>Manly</snm><fnm>B. F. J.</fnm></au>
  </aug>
  <publisher>Boca Raton, FL: Chapman \& Hall / CRC Press</publisher>
  <pubdate>2006</pubdate>
</bibl>

<bibl id="B7">
  <title><p>An Introduction to Behavioural Ecology</p></title>
  <aug>
    <au><snm>Davies</snm><fnm>NB</fnm></au>
    <au><snm>Krebs</snm><fnm>JR</fnm></au>
  </aug>
  <publisher>U.S.A.: Wiley-Blackwell</publisher>
  <pubdate>1993</pubdate>
</bibl>

<bibl id="B8">
  <title><p>Essentials of Ecology</p></title>
  <aug>
    <au><snm>Townsend</snm><fnm>CR</fnm></au>
    <au><snm>Begon</snm><fnm>M</fnm></au>
    <au><snm>Harper</snm><fnm>JL</fnm></au>
  </aug>
  <publisher>U.S.A.: Wiley-Blackwell</publisher>
  <pubdate>2002</pubdate>
</bibl>

<bibl id="B9">
  <title><p>Vegetation of the Siskiyou Mountains, Oregon and
  California</p></title>
  <aug>
    <au><snm>Whittaker</snm><fnm>R. H.</fnm></au>
  </aug>
  <source>Ecological Monographs</source>
  <publisher>Wiley-Blackwell</publisher>
  <pubdate>1960</pubdate>
  <volume>30</volume>
  <issue>3</issue>
  <fpage>279</fpage>
  <lpage>-338</lpage>
</bibl>

<bibl id="B10">
  <title><p>Beta Diversity on Geographic Gradients in Britain</p></title>
  <aug>
    <au><snm>Harrison</snm><fnm>S</fnm></au>
    <au><snm>Ross</snm><fnm>SJ</fnm></au>
    <au><snm>Lawton</snm><fnm>JH</fnm></au>
  </aug>
  <source>The Journal of Animal Ecology</source>
  <publisher>{JSTOR}</publisher>
  <pubdate>1992</pubdate>
  <volume>61</volume>
  <issue>1</issue>
  <fpage>151</fpage>
</bibl>

<bibl id="B11">
  <title><p>Measuring beta diversity for presence-absence data</p></title>
  <aug>
    <au><snm>Koleff</snm><fnm>P</fnm></au>
    <au><snm>Gaston</snm><fnm>KJ</fnm></au>
    <au><snm>Lennon</snm><fnm>JJ</fnm></au>
  </aug>
  <source>Journal of Animal Ecology</source>
  <publisher>Wiley-Blackwell</publisher>
  <pubdate>2003</pubdate>
  <volume>72</volume>
  <issue>3</issue>
  <fpage>367</fpage>
  <lpage>-382</lpage>
</bibl>

<bibl id="B12">
  <title><p>The Assembly of Species Communities: Chance or
  Competition?</p></title>
  <aug>
    <au><snm>Connor</snm><fnm>EF</fnm></au>
    <au><snm>Simberloff</snm><fnm>D</fnm></au>
  </aug>
  <source>Ecology</source>
  <publisher>Wiley-Blackwell</publisher>
  <pubdate>1979</pubdate>
  <volume>60</volume>
  <issue>6</issue>
  <fpage>1132</fpage>
  <url>http://dx.doi.org/10.2307/1936961</url>
</bibl>

<bibl id="B13">
  <title><p>Examination of the "null" model of Connor and Simberloff for
  species co-occurrence on islands</p></title>
  <aug>
    <au><snm>Diamond</snm><fnm>J. M.</fnm></au>
    <au><snm>Gilpin</snm><fnm>M. E.</fnm></au>
  </aug>
  <source>Oecologia</source>
  <pubdate>1982</pubdate>
  <volume>52</volume>
  <fpage>64</fpage>
  <lpage>74</lpage>
</bibl>

<bibl id="B14">
  <title><p>Factors Contributing to Non-Randomness in Species Co-Occurrences on
  Islands</p></title>
  <aug>
    <au><snm>Gilpin</snm><fnm>M. E.</fnm></au>
    <au><snm>Diamond</snm><fnm>J. M.</fnm></au>
  </aug>
  <source>Oecologia</source>
  <pubdate>1982</pubdate>
  <volume>52</volume>
  <fpage>75</fpage>
  <lpage>84</lpage>
</bibl>

<bibl id="B15">
  <title><p>Methods for detecting non-randomness in species co-occurrences: a
  contribution</p></title>
  <aug>
    <au><snm>Wilson</snm><fnm>J.B.</fnm></au>
  </aug>
  <source>Oecologia</source>
  <pubdate>1987</pubdate>
  <volume>73</volume>
  <issue>4</issue>
  <fpage>579</fpage>
  <lpage>-582</lpage>
</bibl>

<bibl id="B16">
  <title><p>A Note on the Analysis of Species Co-Occurrences</p></title>
  <aug>
    <au><snm>Manly</snm><fnm>B. F. J.</fnm></au>
  </aug>
  <source>Ecology</source>
  <pubdate>1995</pubdate>
  <volume>76</volume>
  <issue>4</issue>
  <fpage>1109</fpage>
  <lpage>-1115</lpage>
</bibl>

<bibl id="B17">
  <title><p>Null matrices and the analysis of species
  co-occurrencesSanderson</p></title>
  <aug>
    <au><snm>Sanderson</snm><fnm>J.</fnm></au>
    <au><snm>Moulton</snm><fnm>M.</fnm></au>
    <au><snm>Selfridge</snm><fnm>R.</fnm></au>
  </aug>
  <source>Oecologia</source>
  <pubdate>1998</pubdate>
  <volume>116</volume>
  <issue>1--2</issue>
  <fpage>275</fpage>
  <lpage>-283</lpage>
</bibl>

<bibl id="B18">
  <title><p>Stochastic and deterministic processes jointly structure tropical
  arthropod communities</p></title>
  <aug>
    <au><snm>Ellwood</snm><fnm>MDF</fnm></au>
    <au><snm>Manica</snm><fnm>A</fnm></au>
    <au><snm>Foster</snm><fnm>WA</fnm></au>
  </aug>
  <source>Ecology Letters</source>
  <publisher>Wiley-Blackwell</publisher>
  <pubdate>2009</pubdate>
  <volume>12</volume>
  <issue>4</issue>
  <fpage>277</fpage>
  <lpage>-284</lpage>
</bibl>

<bibl id="B19">
  <title><p>Disentangling the importance of ecological niches from stochastic
  processes across scales</p></title>
  <aug>
    <au><snm>Chase</snm><fnm>JM</fnm></au>
    <au><snm>Myers</snm><fnm>JA</fnm></au>
  </aug>
  <source>Philosophical transactions of the Royal Society B: Biological
  sciences</source>
  <publisher>The Royal Society</publisher>
  <pubdate>2011</pubdate>
  <volume>366</volume>
  <issue>1576</issue>
  <fpage>2351</fpage>
  <lpage>-2363</lpage>
</bibl>

<bibl id="B20">
  <title><p>Co-occurrence based assessment of habitat generalists and
  specialists: A new approach for the measurement of niche width</p></title>
  <aug>
    <au><snm>Fridley</snm><fnm>JD</fnm></au>
    <au><snm>Vandermast</snm><fnm>DB</fnm></au>
    <au><snm>Kuppinger</snm><fnm>DM</fnm></au>
    <au><snm>Manthey</snm><fnm>M</fnm></au>
    <au><snm>Peet</snm><fnm>RK</fnm></au>
  </aug>
  <source>Journal of Ecology</source>
  <publisher>Wiley Online Library</publisher>
  <pubdate>2007</pubdate>
  <volume>95</volume>
  <issue>4</issue>
  <fpage>707</fpage>
  <lpage>-722</lpage>
</bibl>

<bibl id="B21">
  <title><p>The geographic scaling of biotic interactions</p></title>
  <aug>
    <au><snm>Ara{\'{u}}jo</snm><fnm>MB</fnm></au>
    <au><snm>Rozenfeld</snm><fnm>A</fnm></au>
  </aug>
  <source>Ecography</source>
  <publisher>Wiley-Blackwell</publisher>
  <pubdate>2013</pubdate>
</bibl>

<bibl id="B22">
  <title><p>Similarity of Binary Data</p></title>
  <aug>
    <au><snm>Baroni Urbani</snm><fnm>C</fnm></au>
    <au><snm>Buser</snm><fnm>MW</fnm></au>
  </aug>
  <source>Systematic Zoology</source>
  <publisher>Oxford University Press ({OUP})</publisher>
  <pubdate>1976</pubdate>
  <volume>25</volume>
  <issue>3</issue>
  <fpage>251</fpage>
</bibl>

<bibl id="B23">
  <title><p>A statistical table for the degree of coexistence between two
  species</p></title>
  <aug>
    <au><snm>Baroni Urbani</snm><fnm>C</fnm></au>
  </aug>
  <source>Oecologia</source>
  <publisher>Springer Nature</publisher>
  <pubdate>1979</pubdate>
  <volume>44</volume>
  <issue>3</issue>
  <fpage>287</fpage>
  <lpage>-289</lpage>
  <url>https://doi.org/10.1007%2Fbf00545229</url>
</bibl>

<bibl id="B24">
  <title><p>A probabilistic model for analysing species
  co-occurrence</p></title>
  <aug>
    <au><snm>Veech</snm><fnm>JA</fnm></au>
  </aug>
  <source>Global Ecology and Biogeography</source>
  <pubdate>2013</pubdate>
  <volume>22</volume>
  <fpage>252</fpage>
  <lpage>260</lpage>
</bibl>

<bibl id="B25">
  <title><p>cooccur: Probabilistic Species Co-Occurrence Analysis
  inR</p></title>
  <aug>
    <au><snm>Griffith</snm><fnm>DM</fnm></au>
    <au><snm>Veech</snm><fnm>JA</fnm></au>
    <au><snm>Marsh</snm><fnm>CJ</fnm></au>
  </aug>
  <source>Journal of Statistical Software</source>
  <pubdate>2016</pubdate>
  <volume>69</volume>
</bibl>

<bibl id="B26">
  <title><p>R: A Language and Environment for Statistical Computing</p></title>
  <aug>
    <au><cnm>{R Core Team}</cnm></au>
  </aug>
  <publisher>Vienna, Austria</publisher>
  <pubdate>2017</pubdate>
  <url>https://www.R-project.org</url>
</bibl>

<bibl id="B27">
  <title><p>All of Statistics: A Concise Course in Statistical
  Inference</p></title>
  <aug>
    <au><snm>Wasserman</snm><fnm>L.</fnm></au>
  </aug>
  <publisher>New York, U.S.A.: Springer</publisher>
  <pubdate>2010</pubdate>
</bibl>

<bibl id="B28">
  <title><p>{IsoSpec}: Hyperfast Fine Structure Calculator</p></title>
  <aug>
    <au><snm>{\L}{\k{a}}cki</snm><fnm>MK</fnm></au>
    <au><snm>Startek</snm><fnm>M</fnm></au>
    <au><snm>Valkenborg</snm><fnm>D</fnm></au>
    <au><snm>Gambin</snm><fnm>A</fnm></au>
  </aug>
  <source>Analytical Chemistry</source>
  <publisher>American Chemical Society ({ACS})</publisher>
  <pubdate>2017</pubdate>
  <volume>89</volume>
  <issue>6</issue>
  <fpage>3272</fpage>
  <lpage>3277</lpage>
  <url>https://doi.org/10.1021%2Facs.analchem.6b01459</url>
</bibl>

<bibl id="B29">
  <title><p>An Introduction to the Bootstrap</p></title>
  <aug>
    <au><snm>Efron</snm><fnm>B.</fnm></au>
    <au><snm>Tibshirani</snm><fnm>R.</fnm></au>
  </aug>
  <publisher>Boca Raton, Florida: Chapman \& Hall / CRC Press</publisher>
  <pubdate>1994</pubdate>
</bibl>

<bibl id="B30">
  <title><p>Species Number and Compositional Similarity of the Galapagos Flora
  and Avifauna</p></title>
  <aug>
    <au><snm>Connor</snm><fnm>E. F.</fnm></au>
    <au><snm>Simberloff</snm><fnm>D</fnm></au>
  </aug>
  <source>Ecological Monographs</source>
  <pubdate>1978</pubdate>
  <volume>48</volume>
  <fpage>219</fpage>
  <lpage>248</lpage>
</bibl>

<bibl id="B31">
  <title><p>EcoSimR: Null model analysis for ecological data</p></title>
  <aug>
    <au><snm>Gotelli</snm><fnm>NJ</fnm></au>
    <au><snm>Hart</snm><fnm>EM</fnm></au>
    <au><snm>Ellison</snm><fnm>AM</fnm></au>
  </aug>
  <pubdate>2015</pubdate>
  <url>http://github.com/gotellilab/EcoSimR</url>
  <note>R package version 0.1.0</note>
</bibl>

<bibl id="B32">
  <title><p>vegan: Community Ecology Package</p></title>
  <aug>
    <au><snm>Oksanen</snm><fnm>J</fnm></au>
    <au><snm>Blanchet</snm><fnm>FG</fnm></au>
    <au><snm>Friendly</snm><fnm>M</fnm></au>
    <au><snm>Kindt</snm><fnm>R</fnm></au>
    <au><snm>Legendre</snm><fnm>P</fnm></au>
    <au><snm>McGlinn</snm><fnm>D</fnm></au>
    <au><snm>Minchin</snm><fnm>PR</fnm></au>
    <au><snm>O'Hara</snm><fnm>R. B.</fnm></au>
    <au><snm>Simpson</snm><fnm>GL</fnm></au>
    <au><snm>Solymos</snm><fnm>P</fnm></au>
    <au><snm>Stevens</snm><fnm>MHH</fnm></au>
    <au><snm>Szoecs</snm><fnm>E</fnm></au>
    <au><snm>Wagner</snm><fnm>H</fnm></au>
  </aug>
  <pubdate>2017</pubdate>
  <url>https://CRAN.R-project.org/package=vegan</url>
  <note>R package version 2.4-5</note>
</bibl>

<bibl id="B33">
  <title><p>Statistical significance for genomewide studies</p></title>
  <aug>
    <au><snm>Storey</snm><fnm>J. D.</fnm></au>
    <au><snm>Tibshirani</snm><fnm>R.</fnm></au>
  </aug>
  <source>Proceedings of the National Academy of Sciences</source>
  <publisher>National Acad Sciences</publisher>
  <pubdate>2003</pubdate>
  <volume>100</volume>
  <issue>16</issue>
  <fpage>9440</fpage>
  <lpage>9445</lpage>
</bibl>

<bibl id="B34">
  <title><p>Climate interacts with anthropogenic drivers to determine
  extirpation dynamics</p></title>
  <aug>
    <au><snm>Comte</snm><fnm>L</fnm></au>
    <au><snm>Hugueny</snm><fnm>B</fnm></au>
    <au><snm>Grenouillet</snm><fnm>G</fnm></au>
  </aug>
  <source>Ecography</source>
  <publisher>Wiley Online Library</publisher>
  <pubdate>2016</pubdate>
  <volume>39</volume>
  <issue>10</issue>
  <fpage>1008</fpage>
  <lpage>-1016</lpage>
</bibl>

<bibl id="B35">
  <title><p>Similarity Coefficients for Binary Chemoinformatics Data: Overview
  and Extended Comparison Using Simulated and Real Data Sets</p></title>
  <aug>
    <au><snm>Todeschini</snm><fnm>R</fnm></au>
    <au><snm>Consonni</snm><fnm>V</fnm></au>
    <au><snm>Xiang</snm><fnm>H</fnm></au>
    <au><snm>Holliday</snm><fnm>J</fnm></au>
    <au><snm>Buscema</snm><fnm>M</fnm></au>
    <au><snm>Willett</snm><fnm>P</fnm></au>
  </aug>
  <source>J. Chem. Inf. Model.</source>
  <publisher>American Chemical Society ({ACS})</publisher>
  <pubdate>2012</pubdate>
  <volume>52</volume>
  <issue>11</issue>
  <fpage>2884</fpage>
  <lpage>-2901</lpage>
  <url>http://dx.doi.org/10.1021/ci300261r</url>
</bibl>

<bibl id="B36">
  <title><p>{EC}-{BLAST}: a tool to automatically search and compare enzyme
  reactions</p></title>
  <aug>
    <au><snm>Rahman</snm><fnm>SA</fnm></au>
    <au><snm>Cuesta</snm><fnm>SM</fnm></au>
    <au><snm>Furnham</snm><fnm>N</fnm></au>
    <au><snm>Holliday</snm><fnm>GL</fnm></au>
    <au><snm>Thornton</snm><fnm>JM</fnm></au>
  </aug>
  <source>Nature Methods</source>
  <publisher>Springer Nature</publisher>
  <pubdate>2014</pubdate>
  <volume>11</volume>
  <issue>2</issue>
  <fpage>171</fpage>
  <lpage>-174</lpage>
  <url>https://doi.org/10.1038%2Fnmeth.2803</url>
</bibl>

<bibl id="B37">
  <title><p>Why is Tanimoto index an appropriate choice for fingerprint-based
  similarity calculations?</p></title>
  <aug>
    <au><snm>Bajusz</snm><fnm>D</fnm></au>
    <au><snm>R{\'{a}}cz</snm><fnm>A</fnm></au>
    <au><snm>H{\'{e}}berger</snm><fnm>K</fnm></au>
  </aug>
  <source>J Cheminform</source>
  <publisher>Springer Science $\mathplus$ Business Media</publisher>
  <pubdate>2015</pubdate>
  <volume>7</volume>
  <issue>1</issue>
  <url>http://dx.doi.org/10.1186/s13321-015-0069-3</url>
</bibl>

<bibl id="B38">
  <title><p>BEDTools: the Swiss-army tool for genome feature
  analysis</p></title>
  <aug>
    <au><snm>Quinlan</snm><fnm>AR</fnm></au>
  </aug>
  <source>Current Protocols in Bioinformatics</source>
  <publisher>Wiley Online Library</publisher>
  <pubdate>2014</pubdate>
  <fpage>11</fpage>
  <lpage>-12</lpage>
</bibl>

</refgrp>
} % end of \BMCxmlcomment
\clearpage
\section*{Figures}

\begin{figure}[hbt]
\begin{center}
    \includegraphics[width=.5\textwidth]{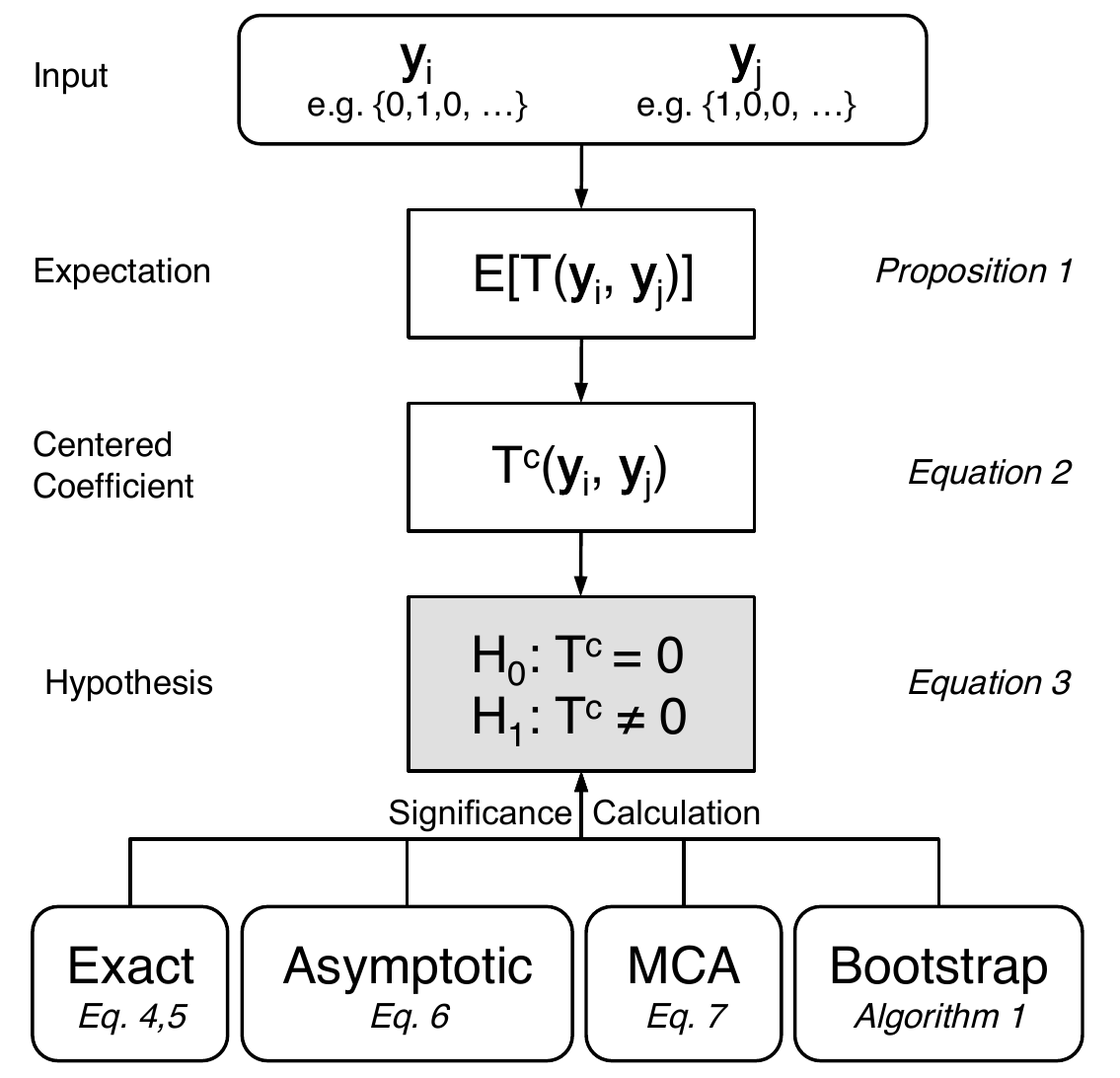}
	\caption{\small Flowchart of the proposed statistical methods and algorithms.}
	\label{flowchart}
\end{center}
\end{figure}

\begin{figure}[hbt]
\begin{center}
    \includegraphics[width=.9\textwidth]{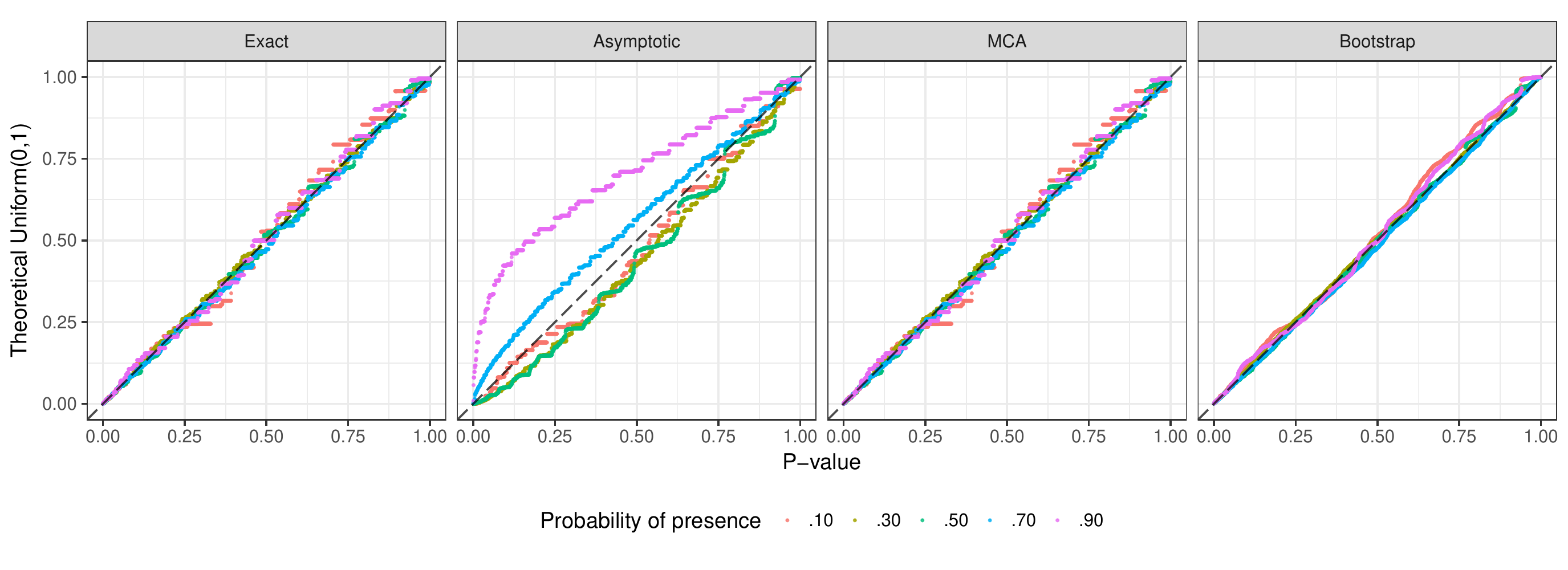}
	\caption{\small P-values of similarity among independent presence-absence vectors of $m=100$, with a wide range of probabilities $p=.1,.3,.5,.7,.9$. In each scenario, 2000 independent variables are simulated and tested using four proposed methods. The diagonal lines indicate a theoretically correct Uniform(0,1) distribution.}
	\label{NullOnly_QQ_nvar100pAll}
\end{center}
\end{figure}

\begin{figure}[ht]
\begin{center}
    \includegraphics[width=.9\textwidth]{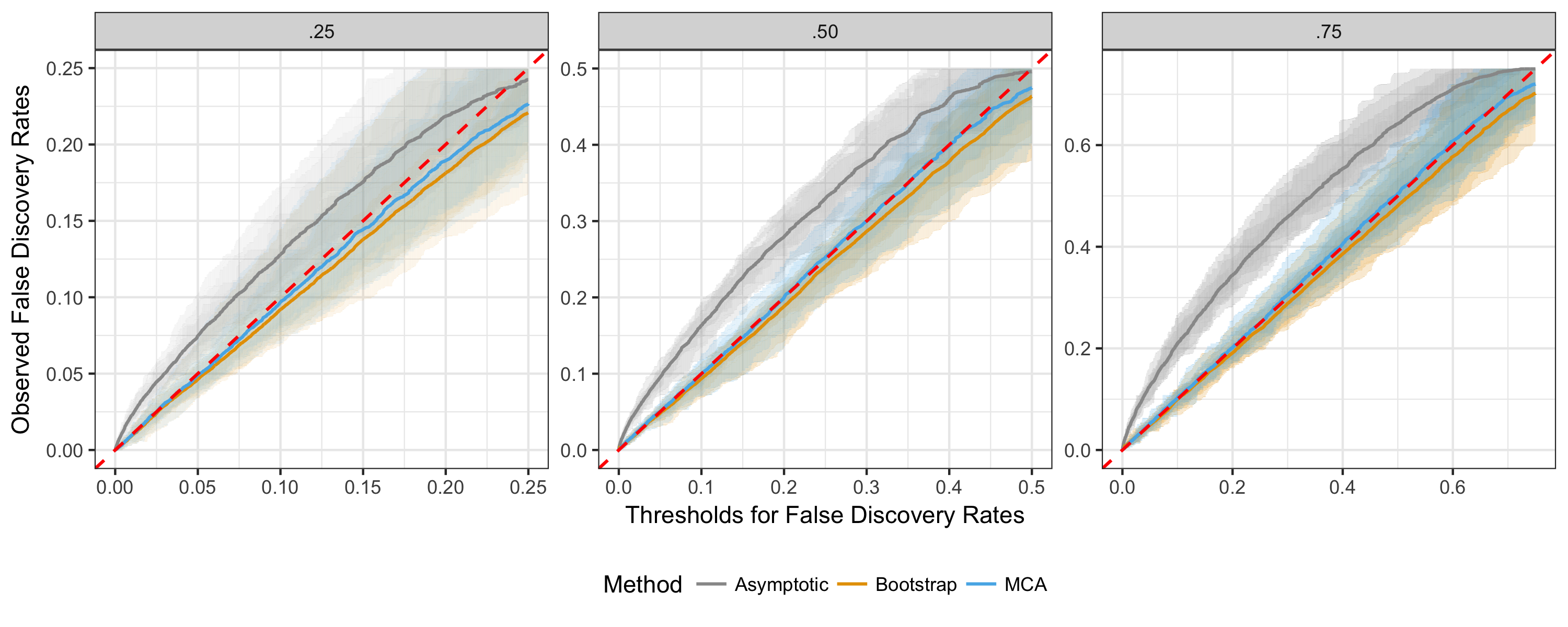}
	\caption{\small False discovery rate (FDR) estimates from a mixture of independent and dependent presence-absence vectors. In 3 separate scenarios with null proportions $\pi_0 = .25, .50, .75$, 2000 presence-absence vectors of $m=200$ are simulated with occurrence probabilities of $p=.5$. Each simulation scenario is repeated 20 times and the proposed methods are used to automatically compute p-values and q-values. FDR thresholds are plotted against observed false discovery proportions, where a downward deviation from a theoretically correct diagonal red line indicates a conservative behavior.}
	\label{FDR_nvar200}
\end{center}
\end{figure}

\begin{figure}[ht]
\begin{center}
    \includegraphics[width=.6\textwidth]{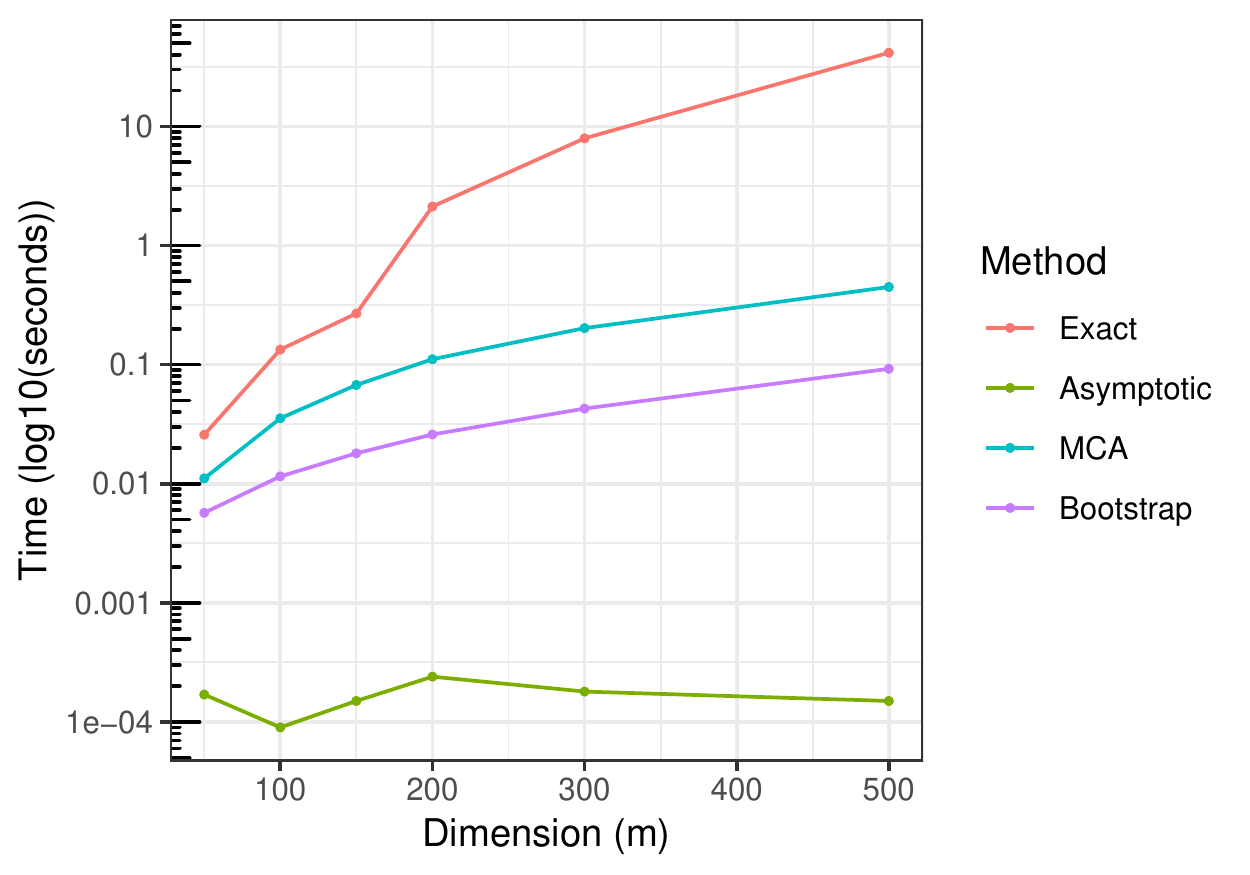}
	\caption{\small Computational runtimes of our 4 proposed methods. The means of 100 independent runs are plotted against an increasing size of dimension $m= 50, \ldots, 500$. Compared to the exact solution, the bootstrap and measure concentration algorithm (MCA) provide vast improvements in speed whose relative efficiency increases with higher dimension.}
	\label{runtime_main}
\end{center}
\end{figure}

\begin{figure}[t]
\begin{center}
    \includegraphics[width=.5\textwidth]{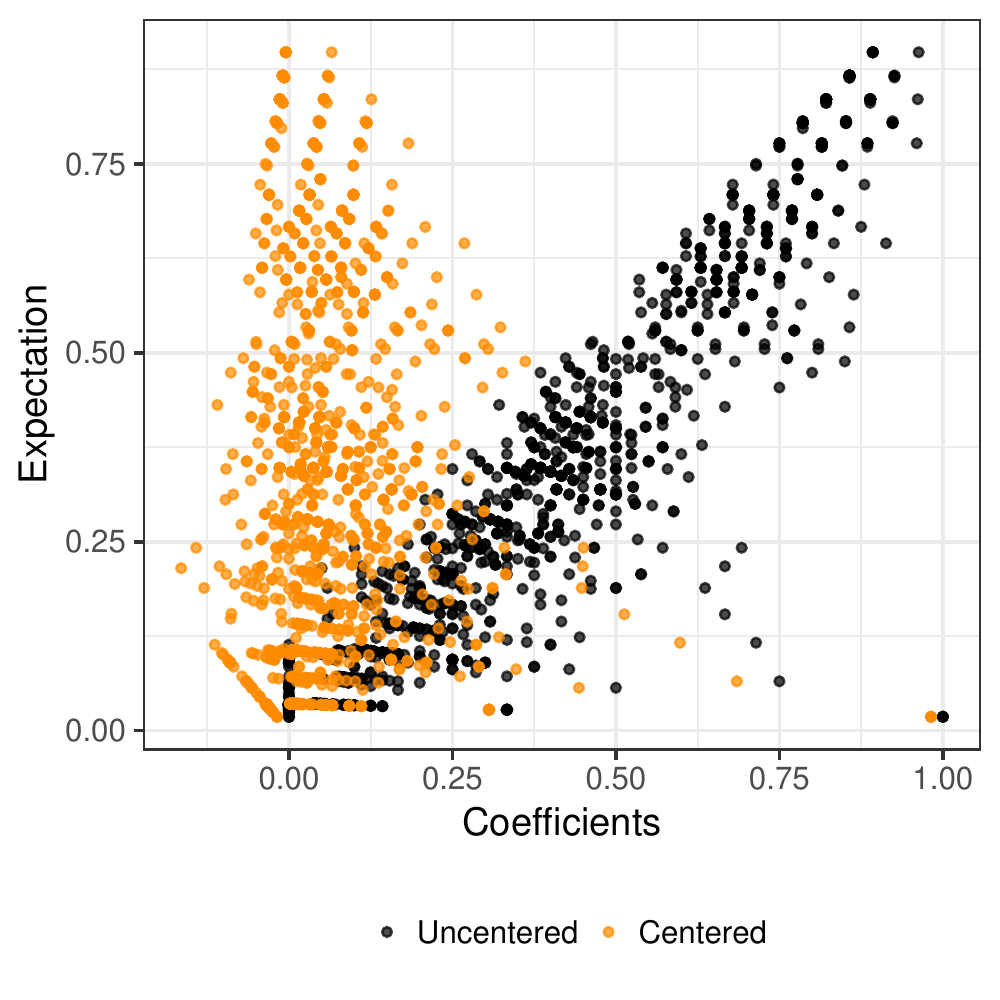}
	\caption{\small Comparison of uncentered and centered Jaccard/Tanimoto coefficients from the bird dataset. The conventional uncentered coefficients are shown to be strongly dependent on expectation under independence. By centering each coefficient by its expectation, the proposed centered coefficients alleviate this dependency.}
	\label{compare_centering}
\end{center}
\end{figure}

\begin{figure}[t]
\begin{center}
    \includegraphics[width=.8\textwidth]{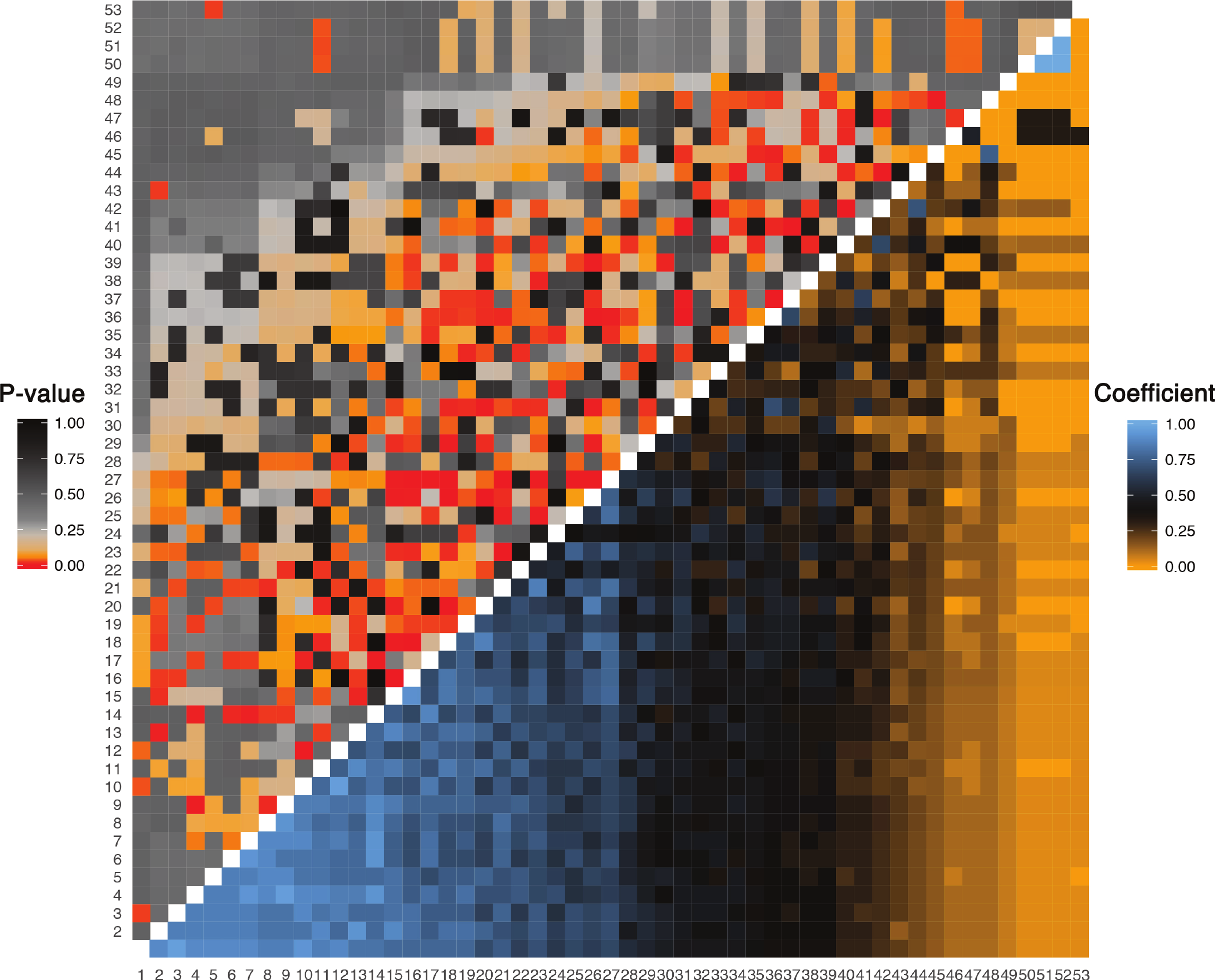}
	\caption{\small Heatmap of uncentered Jaccard/Tanimoto coefficients and their p-values. Similarity among 53 bird species in 28 islands of Vanuatu are tested using the proposed method. Species are ordered from high to low occurrences, that are highly correlated with Jaccard/Tanimoto coefficients (p-value $< 2.2 \times 10^{-16}$). The upper triangle shows the p-values from our methods, whereas the lower triangle the observed Jaccard/Tanimoto coefficients.}
	\label{bird_heatmap_combined}
\end{center}
\end{figure}

%%%%%%%%%%%%%%%%%%%%%%%%%%%%%%%%%%%
%%                               %%
%% Tables                        %%
%%                               %%
%%%%%%%%%%%%%%%%%%%%%%%%%%%%%%%%%%%

%%%%%%%%%%%%%%%%%%%%%%%%%%%%%%%%%%%
%%                               %%
%% Additional Files              %%
%%                               %%
%%%%%%%%%%%%%%%%%%%%%%%%%%%%%%%%%%%

\clearpage
\beginsupplement
\section*{Supplementary Figures} 
\begin{figure}[ht]
\begin{center}
    \includegraphics[width=.9\textwidth]{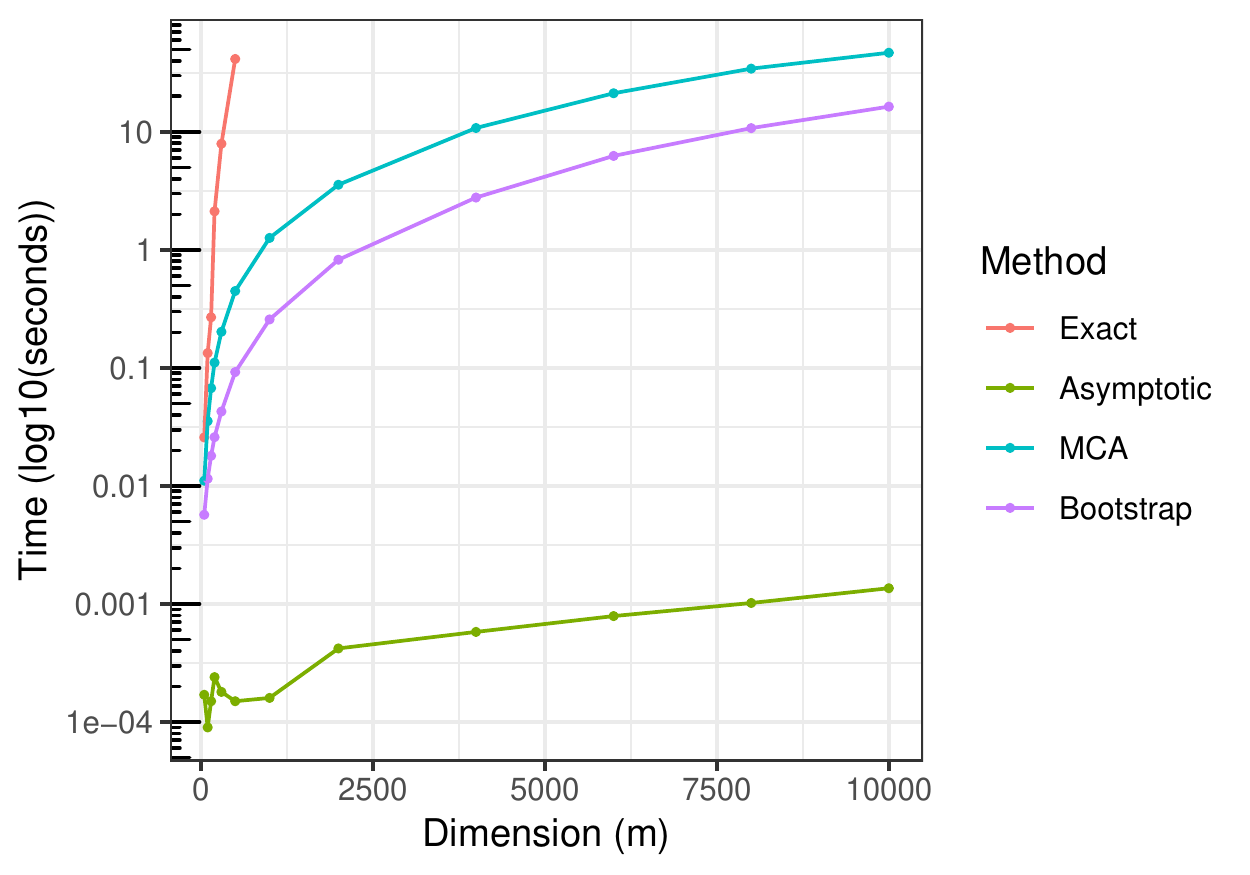}
	\caption{\small Computational runtimes when testing similarity between presence-absence data upto $m=10000$. We ran the proposed 4 methods to compute p-values for a wide range of dimension $m$. For each $m$, 100 independent simulations are conducted. Note that for $m \geq 1000$, the exact solution did not compute in a reasonable time. The bootstrap and measure concentration algorithm (MCA) are orders of magnitude faster than the exact solution. The asymptotic solution is instantaneous regardless of $m$.   }
	\label{runtime_supp}
\end{center}
\end{figure}

\begin{figure}[ht]
\begin{center}
    \includegraphics[width=.8\textwidth]{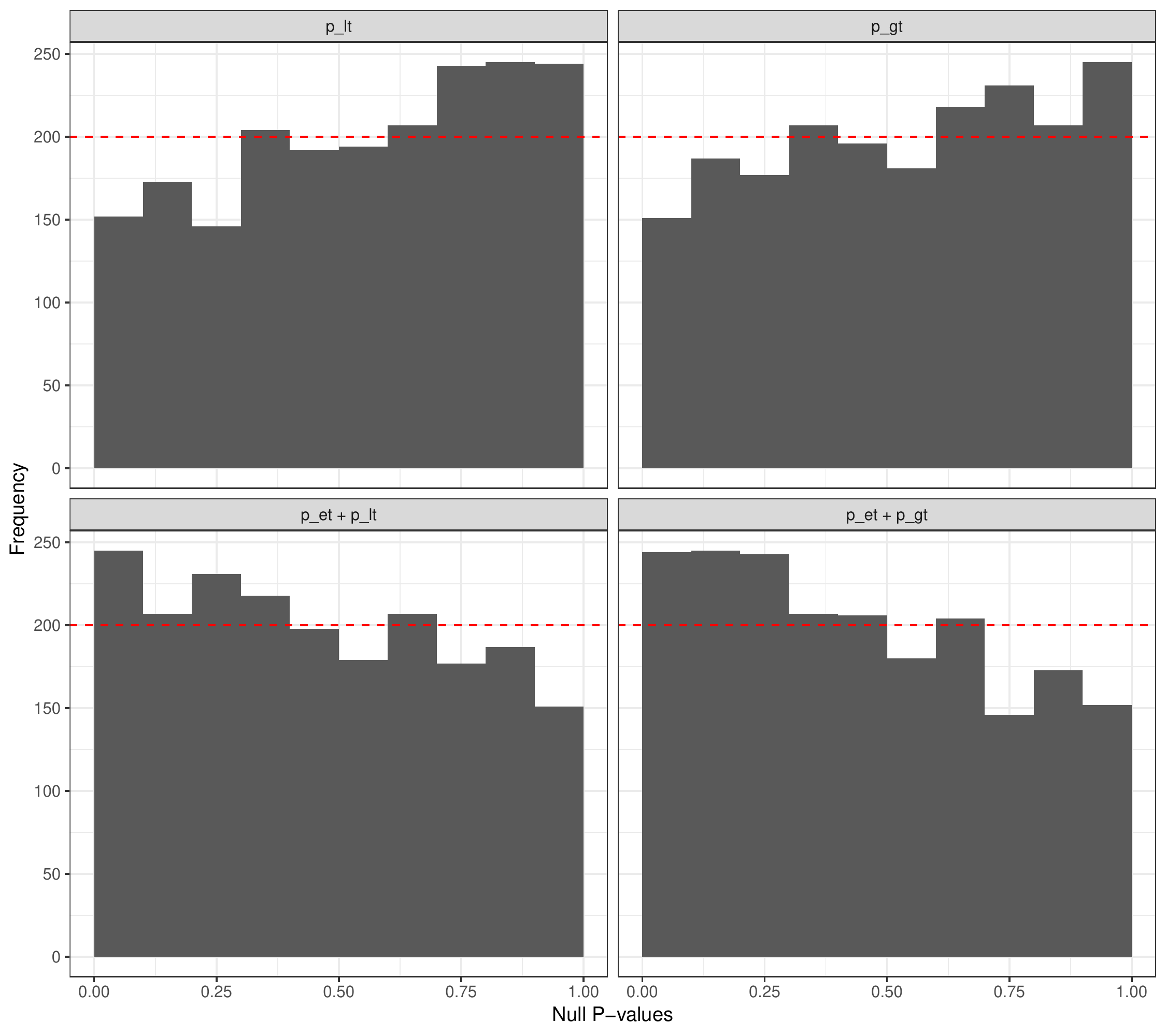}
	\caption{\small Combinatoric p-values of similarity among independent presence-absence vectors of $m=200$ with $p=.5$. In each scenario, 2000 independent variables are simulated and tested using a combinatorics \cite{Veech2013}. \cite{Veech2013} recommends $\textnormal{p}_\textnormal{lt} + \textnormal{p}_\textnormal{et}$ and $\textnormal{p}_\textnormal{gt} + \textnormal{p}_\textnormal{et}$ as p-values. The dashed red lines indicate theoretically correct Uniform distributions.}
	\label{Veech_NullPvalues}
\end{center}
\end{figure}

\begin{figure}[ht]
\begin{center}
    \includegraphics[width=.8\textwidth]{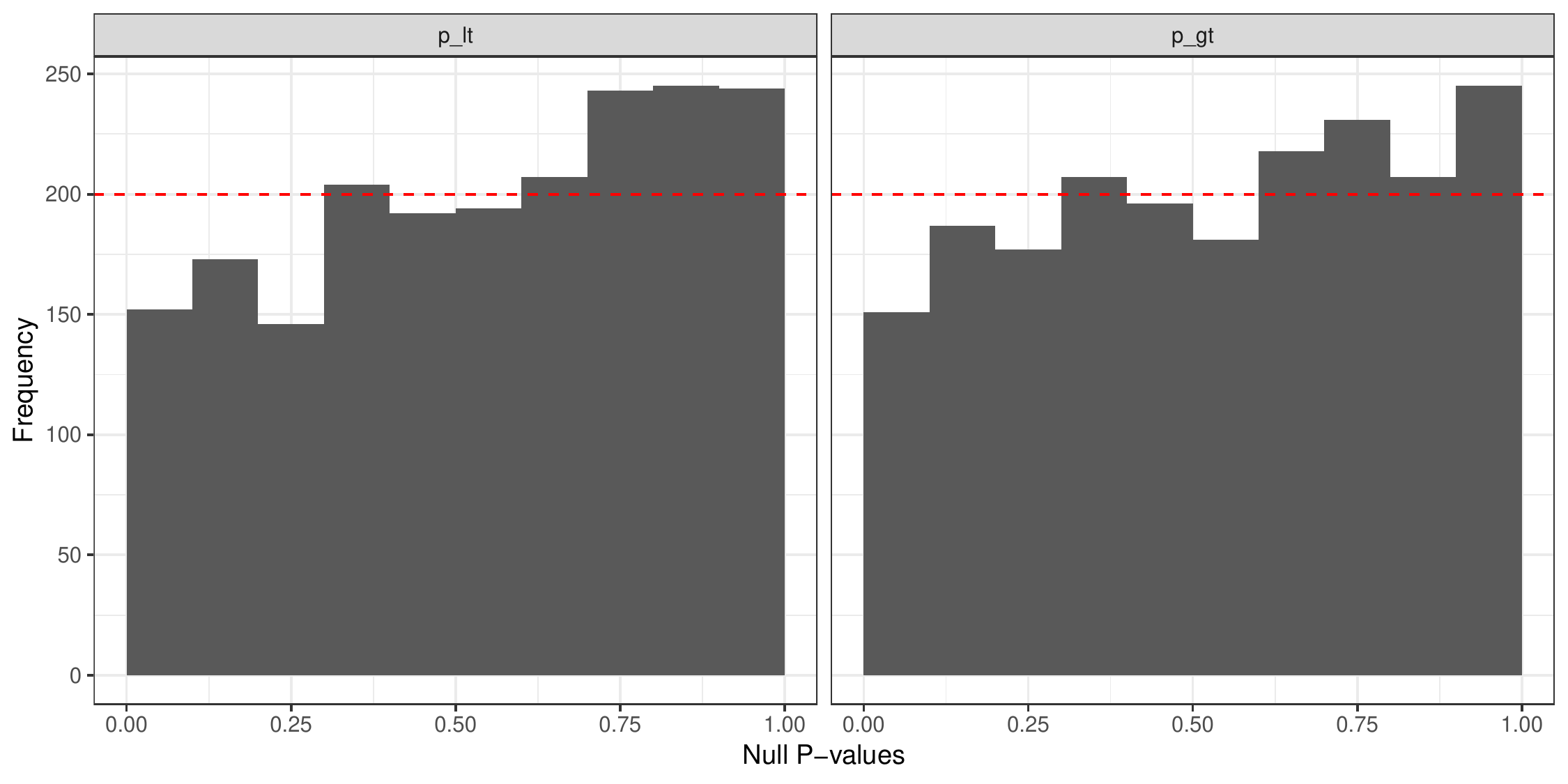}
	\caption{\small Hypergeometric p-values of similarity among independent presence-absence vectors of $m=200$ with $p=.5$. We used a hypergeometric distribution \cite{Griffith2016} to obtain p-values of similarity between independent species. The original authors suggested that $\textnormal{p}_\textnormal{gt}$ and $\textnormal{p}_\textnormal{lt}$ can be ``interpreted and reported as p-values''. The dashed red lines indicate theoretically correct Uniform distributions.}
	\label{Cooccur_NullPvalues}
\end{center}
\end{figure}

\begin{figure}[h]
\begin{center}
    \includegraphics[width=.8\textwidth]{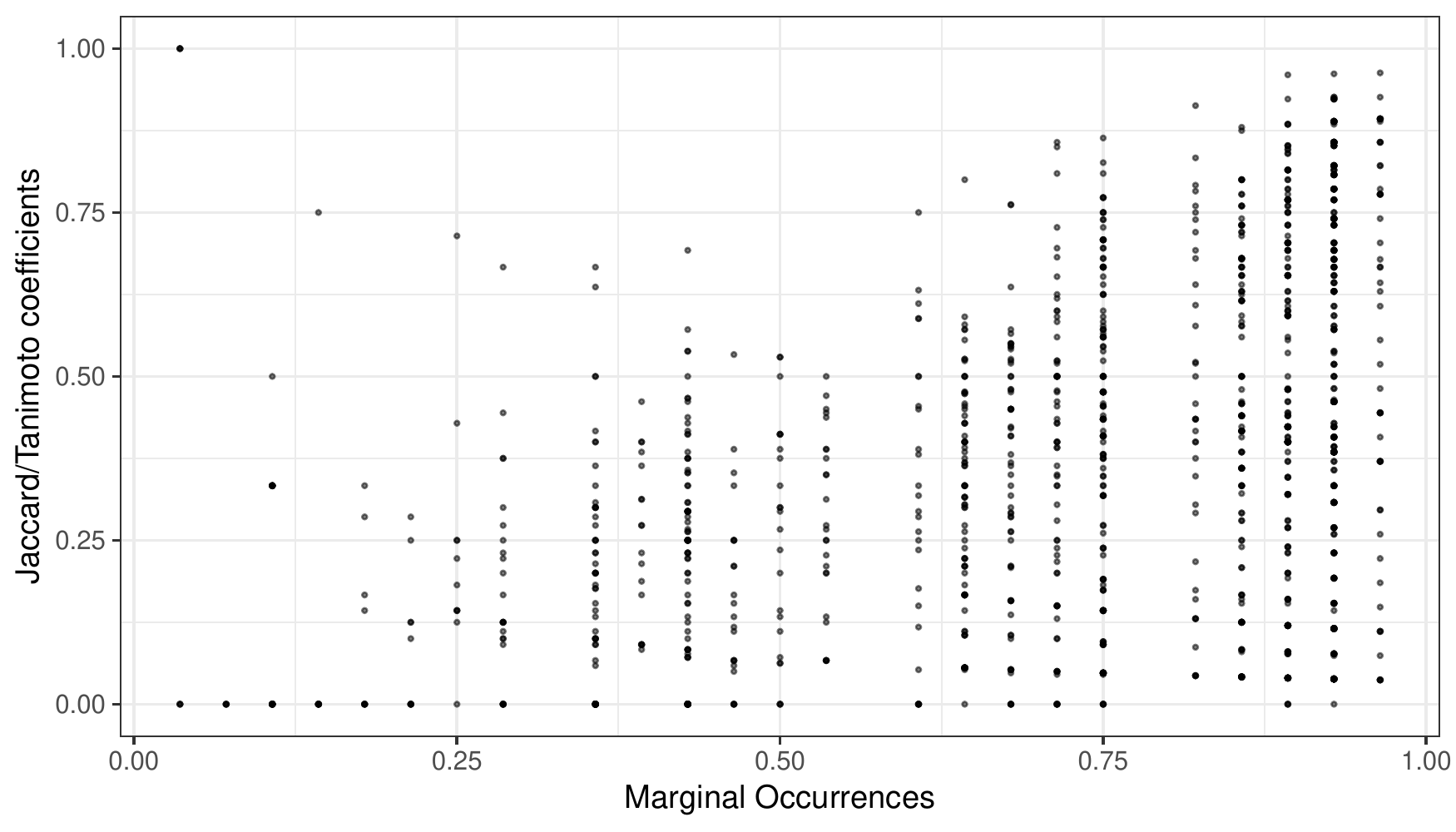}
	\caption{\small Scatterplot of marginal occurrences of 53 bird species and Jaccard/Tanimoto coefficients. As expected, we observe high correlation (Pearson correlation $= 0.43$) between marginal occurrences and Jaccard/Tanimoto coefficients.}
	\label{bird_occurrences_coefficients}
\end{center}
\end{figure}
	
\begin{figure}[h]
\begin{center}
    \includegraphics[width=.8\textwidth]{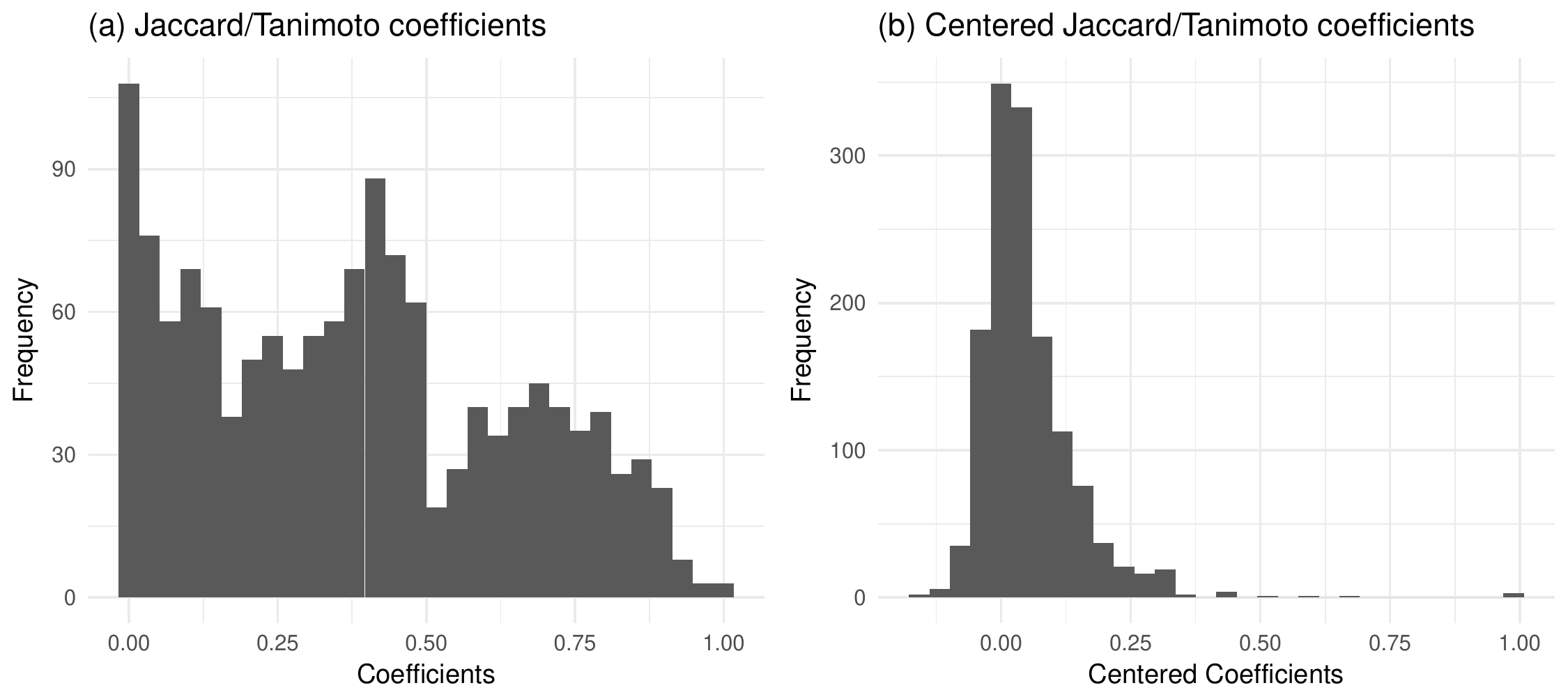}
	\caption{\small Histograms of conventional and centered Jaccard/Tanimoto similarity coefficients. The conventional (uncentered) Jaccard/Tanimoto coefficients are centered by their expected values under the independence assumption.}
	\label{jaccard_histograms}
\end{center}
\end{figure}

\begin{figure}[h]
\begin{center}
    \includegraphics[width=.6\textwidth]{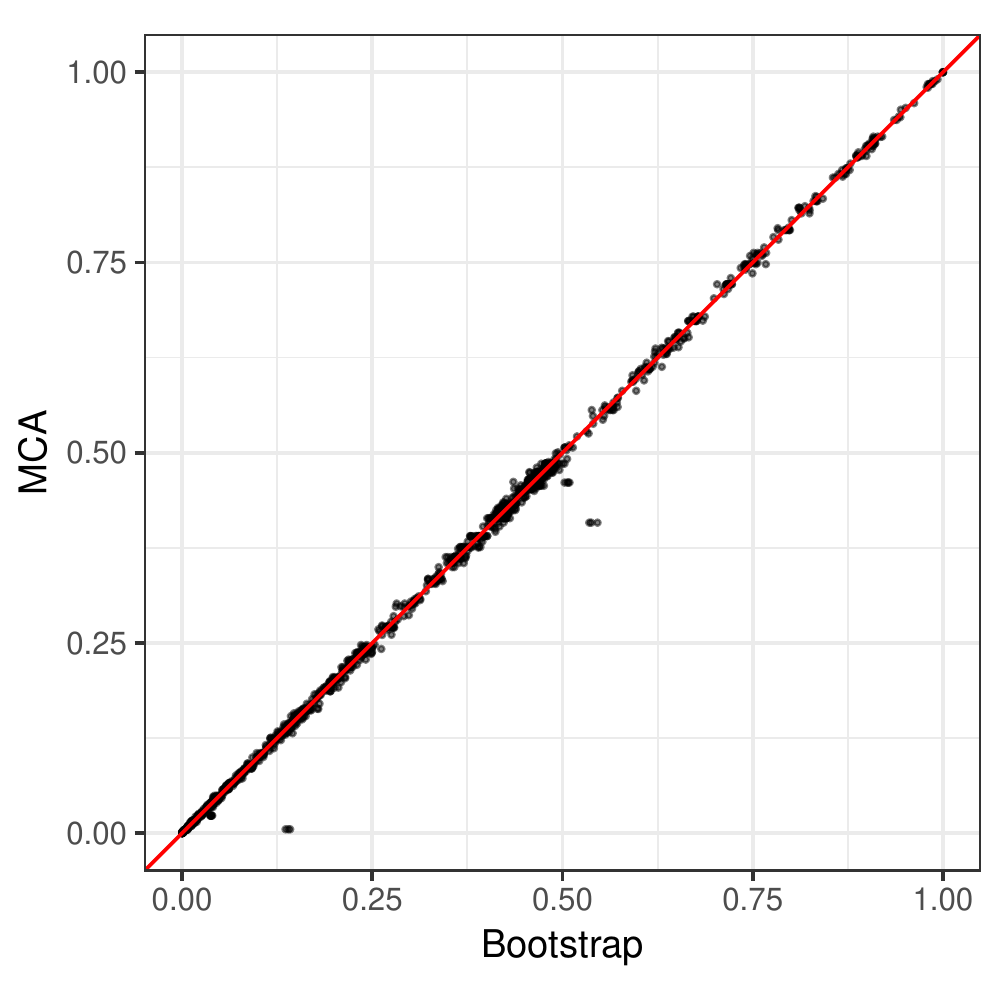}
	\caption{\small Comparison of p-values from the bootstrap and measure concentration algorithm (MCA). Both algorithms were applied on 1378 co-occurrences of bird species. The difference between estimated p-values from two methods is minimal with a mean squared deviation of $1.15 \times 10^{-4}$. The diagonal red line indicates the identity.}
	\label{bird_pvalue_comparison}
\end{center}
\end{figure}

\end{backmatter}
\end{document}